\documentclass[11pt,draftcls]{IEEEtran}
\onecolumn
\usepackage{graphicx} 
\usepackage{float}
\usepackage{amssymb, amsmath,amsthm} 
\usepackage{xcolor}
\usepackage{tikz}
\usetikzlibrary{positioning, arrows.meta}

\newtheorem{theorem}{Theorem}
\newtheorem{lemma}{Lemma}
\newtheorem{definition}{Definition}
\newtheorem{remark}{Remark}

\newcommand{\E}{\mathbb{E}}

\newcommand{\vs}{\mathbf{s}}
\newcommand{\vx}{\mathbf{x}}
\newcommand{\vy}{\mathbf{y}}
\newcommand{\vw}{\mathbf{w}}
\newcommand{\vv}{\mathbf{v}}

\newcommand{\KLQRi}{K_{\mathtt{LQR},i}}
\newcommand{\dhat}[1]{\hat{\hat{#1}}}

\newcommand{\Gti}{\Gamma_{2,i}}
\newcommand{\Gi}{\Gamma_{1,i}}
\newcommand{\sd}{\hat{\Sigma}_i^{\dagger}}
\newcommand{\Err}[1]{(\hat{\vs}_{#1} - \dhat{\vs}_{#1})}
\newcommand{\sumT}{\sum_{i=1}^n}

\DeclareMathOperator{\Tr}{Tr}

\title{Communication over LQG Control Systems: A Convex Optimization Approach to Capacity}
\author{Aharon Rips \quad Oron Sabag
\thanks{The authors are with the Rachel and Selim Benin School of Computer Science and engineering,
Hebrew University of Jerusalem, Israel (\texttt{aharon.rips@mail.huji.ac.il}, \texttt{oron.sabag@mail.huji.ac.il}). The work was supported in part by the Israel Science Foundation (ISF), grant No. 1096/23. }
}

\date{June 2024}

\begin{document}
\maketitle
\begin{abstract}
We study communication over control systems, where a controller-encoder selects inputs to a dynamical system in order to simultaneously regulate the system and convey a message to an observer that has access to the system’s outputs measurements. This setup reflects implicit communication, as the controller embeds a message in the control signal. The capacity of a control system is the maximal reliable rate of the embedded message subject to a closed-loop control-cost constraint. We focus on linear quadratic Gaussian (LQG) control systems, in which the dynamical system is given by a state-space model with Gaussian noise, and the control cost is a quadratic function of the system inputs and system states. Our main result is a convex optimization upper bound on the capacity of LQG systems. In the case of scalar systems, we prove that the upper bound yields the exact LQG system capacity. The upper bound also recovers all known results, including LQG control, feedback capacity of Gaussian channels with memory, and the LQG system capacity with a state-feedback. For vector LQG control systems, we provide a sufficient condition for tightness of the upper bound, based on the Riccati equation. Numerical simulations indicate the upper bound tightness in all tested examples, suggesting that the upper bound may be equal to the LQG system capacity in the vector case as well.
\end{abstract}

\section{Introduction}
Control and communication tasks in engineering systems are traditionally designed independently, each optimized for its own objectives. Control systems concern with regulating the behavior of dynamical systems by adjusting inputs based on feedback, ensuring stability and meeting performance criteria. Communication systems, on the other hand, are designed to reliably transfer information between distinct points over noisy channels, typically aiming at maximizing data rates. Motivated by the decentralized nature of modern systems, we investigate the possible advantages in jointly designing control and communication task on a shared physical system. One way to view this setup is implicit communication where the controller transmits a message to an observer by choosing its actions in the form of a control signal \cite{WuChenGunduz2025Actions}. For example, in a swarm of robots/drones, a leader’s trajectory may act as a control variable to be regulated, but also as an implicit communication signal: followers infer targets by observing the leader’s state through sensors. 



\begin{figure}[t]
    \centering
\begin{tikzpicture}[scale=1.2, 
        every node/.style={font=\footnotesize}, 
        block/.style = {draw, rectangle, minimum height=2em, minimum width=4em, align=center},
        line/.style  = {draw, -{Latex}},
        node distance = 1cm and 1.5cm
    ]

\node[block] (controller) {Controller\\ (Encoder)};
\node[block, right=of controller] (plant) {Plant \\ (Channel)};
\node[above=0.1cm of plant] (plant_label) {Dynamical System}; 
\node[block, right=of plant] (observer) {Observer \\ (Decoder)};
\node[block, below=of plant] (feedback) {Feedback};

\draw[line] (controller.east) -- node[above] {} (plant.west);
\draw[line] (plant.east) -- node[above] {} (observer.west);
\draw[line] (observer.east) -- ++(1,0) node[above] {\hspace{-1cm}Message};
\draw[line] (controller.west) -- ++(-1,0) node[above] {\hspace{1cm}Message} -- (controller.west);
\draw[line] (observer.south) -- (observer.south|-feedback.east) -- (feedback.east);
\draw[line] (feedback.west) -- (controller.south|-feedback.west) -- (controller.south);

\end{tikzpicture}
    \caption{Illustration of the communication-control (CC) problem. A dynamical system serves as both the plant to be controlled and the communication channel. The controller-encoder selects system inputs based the feedback and a message, while the observer-decoder should decode the message based on the system outputs - the measurements. The system operates in a closed-loop configuration, with inputs causally dependent on outputs.}
    \label{fig:control_system}
\end{figure}
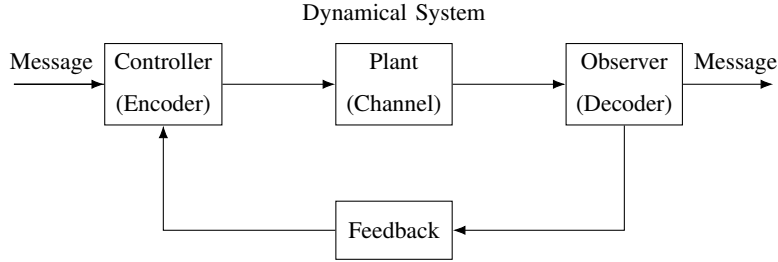

The proposed setting is described in Fig. \ref{fig:control_system}, and is termed the \emph{communication-control} (CC) problem. The main component is a physical environment represented by a stochastic dynamical system. This system describes the plant that needs to be controlled, but is also utilized as a communication channel. A single decision maker, the \emph{controller-encoder}, is choosing the system inputs so as to regulate the control system while embedding an information message that is available to it. Another entity, the \emph{observer-decoder}, has access to the system outputs and needs to decode the message. This operation describes closed-loop systems since the system inputs, chosen by the controller, are causally dependent on the system outputs. The fundamental limits are defined as the maximal achievable message rate subject to a constraint on the closed-loop control cost. We study a fundamental instance of the CC problem: linear Gaussian dynamical systems with a measurement feedback. 



\subsection{The setup}
 The linear quadratic Gaussian (LQG) dynamical system is given by the state-space model
\begin{align}\label{eq:intro_SS}
    \vs_{i+1}&=F \vs_i + G \vx_i + \mathbf{w}_i\nonumber\\
    \vy_i &= H \vs_i + J \vx_i + \mathbf{v}_i,
\end{align}
where $\vs_i \in \mathbb{R}^k$ is the state of the system, $\vx_i\in \mathbb{R}^m$ is the input to the system, and $\vy_i\in \mathbb{R}^p$ is the output of the system, often referred to as the measurement. The random process $(\vw_i,\vv_i)_i$ is independent and identically distributed (i.i.d.) with a Gaussian distribution characterized by $\vw_i \sim \mathcal{N}(0, W)$, $\vv_i \sim \mathcal{N}(0, V)$ and $\E[\vw_i \vv_i^T] = L$. The initial state is distributed as $\vs_1 \sim \mathcal N(0, \Sigma_1)$ and is independent of the sequence $(\vw_i,\vv_i)_i$. The matrices $F\in \mathbb{R}^{k\times k}$, $G\in \mathbb{R}^{k\times m}$, $H\in \mathbb{R}^{p\times k}$ and $J\in \mathbb{R}^{p\times m}$ are known. We proceed to describe the operation of the dynamical system. 

For a fixed horizon (blocklength) $n$ and a rate $R$, the encoder-controller chooses the system input at time $i$, $\vx_i$, as a function of a uniform message $\mathbf{M}\sim U([1:2^{nR}])$ and past system's outputs $\vy^{i-1}$. The observer-decoder has access to the system's measurements and produces a message estimate using a mapping $\widehat{\mathbf{M}} = \phi(\vy^n)$. A policy~$\pi_n$ consists of the controller-encoder and the observer-decoder mappings. The control performance of a policy $\pi_n$ is measured by the linear quadratic regulator (LQR) cost
\begin{align}\label{eq:intro_LQR_cost}
    \mathcal{J}_n&=\frac1{n}\E \left[ \vs_{n+1}^T Q \vs_{n+1} + \sum_{i=1}^n (\vs_i^TQ\vs_i + \vx_i^TR\vx_i) \right],
\end{align}
where $Q \succeq 0, R \succ 0$ are known weight matrices. The LQR control cost reflects a trade-off between regulating the system states and minimizing the control effort (in terms of inputs) required for regulation. Note that the LQR cost is a function of the controller-encoder operation but not a function of the observer-decoder mapping $\phi(\cdot)$.

In the CC Gaussian problem, the \emph{control objective} is to regulate the dynamical system in \eqref{eq:intro_SS} by ensuring that the LQR control cost is below a given control budget. Formally, a policy is $p$-admissible if its LQR cost \eqref{eq:intro_LQR_cost} is smaller than $p$. A rate $R$ is an \emph{achievable rate} if there exists a sequence of $p$-admissible policies $\{\pi_n\}$ such that the probability of error $\Pr(\mathbf{M}\neq \widehat{\mathbf{M}})$ vanishes with~$n$. The capacity of the LQG control system is defined as the supremum over all achievable rates and is denoted by $C_{\mathtt{LQG}}(p)$. Our main objective is to compute the LQG system capacity and to design optimal policies that achieve the capacity.

\begin{remark}
The linear dynamical system in \eqref{eq:intro_SS} is given its most general form. As a communication channel, it has multiple inputs and multiple outputs (MIMO), and the dependence of the system output on a system state allows dependence on past noise instances, e.g., when $H\vs_i +\vv_i$ it is an additive channel with a colored Gaussian noise. The state can also depend on past system inputs and outputs to model intersymbol interference. As a control system, it is also quite general since the controller and the observer only have a measurement feedback, without direct access to the system state, and we do not assume Schur stability of the dynamical system, i.e., $F$ may have eigenvalues outside the unit circle. Standard regularity conditions to guarantee that a finite LQR cost is attainable for large $n$ are provided in Section~\ref{sec:setup}.    
\end{remark}

\subsection{Approach and Related Work}
A special case of the LQG system capacity is feedback capacity of Gaussian channels with memory. In particular, the CC problem is reduced to the communication problem when the LQR cost penalizes the system inputs only (i.e., $Q=0$). In fact, one can utilize existing results in information theory on feedback capacity in order to express the LQG system capacity as an optimization problem of maximizing the directed information. Consider the $n$-letter optimization  
 \begin{align}\label{eq:n_letter_capacity}
     C_n(p)=&\frac{1}{n}\max_{P(\vx^n ||\vy^{n-1}): \mathcal J_n\le p} I(\vx^n \rightarrow \vy^n),
 \end{align}
where $I(\vx^n \rightarrow \vy^n)\triangleq \sumT I(\vx^i;\vy_i|\vy^{i-1})$ is the directed information, and $ P(\vx^n ||\vy^{n-1})\triangleq \prod_{i=1}^n P(\vx_i | \vy^{i-1},\vx^{i-1})$ is the causal conditioning \cite{Massey90,Kramer98}. Then, it follows that $C_{\mathtt{LQG}}(p) =  \lim_{n\to\infty} C_n(p)$. The achievability by Cover and Pombra for colored Gaussian channels with feedback \cite{CoverPombra} can be extended to channels that depend on past channel inputs and outputs and so as to capture the dynamical model in \eqref{eq:intro_SS}, see for instance \cite{sabag22isit}. The converse follows directly, e.g.  \cite{PermuterWeissmanGoldsmith09}, and thus the problem of computing the LQG system capacity is essentially that of computing the limit of the sequence of optimizations in \eqref{eq:n_letter_capacity}. 

The limit over \eqref{eq:n_letter_capacity} is often refereed to as a \emph{multi-letter expression}, and is difficult to compute even for simple dynamical systems. Our main objective is to reformulate this multi-letter expression as a simple computable expression. One relevant class of computable expressions is convex optimization and more specifically semidefinite programming (SDP). Formulating the fundamental limits of control problems as SDPs is very common in the literature \cite{caverly2019lmi}. However, in information theory, there are only few instances where the fundamental limits can be described using SDP and (more precisely, its $\log\det(\cdot)$ variant for the objective function) \cite{tanaka_allerton,Tanaka_SDP_1,Tanaka_SDP_2,Gattami,SabagKostinaHassibiMIMO,Elia_MIMO_ITW,sabag22isit}. If one is able to formulate the LQG system capacity as an SDP, then it can be computed for relatively large systems with tens or even hundreds states using standard software. Before presenting the state of the art for our current setup, we review some of its well-known special cases.



\textbf{Feedback capacity of Gaussian channels with memory:}
Computing the feedback capacity of Gaussian channels with memory has a rich history and indeed is still under research. A representative example for the difficulty is the feedback capacity of the auto-regressive (Markovian) noise process, whose capacity computation had been an open problem for more than $40$ years \cite{Butman69,Butman_conjecture,TiernanSchalk_AR_UB,Ebert,Ordentlich96,YoulaCoding,9611495,Han_GaussianFeedback,elia_bodemeets,Ozarow90} until a simple capacity expression was found \cite{Kim10_Feedback_capacity_stationary_Gaussian}. Indeed, \cite{comments_kim} recently found an error in the derivation of the AR capacity formula in \cite{Kim10_Feedback_capacity_stationary_Gaussian} and therefore a converse proof is missing.

Most related to the current work is the realm of formulation of the feedback capacity as convex optimization. This approach was introduced in \cite{Gattami}, and was later extended to MIMO channels, possibly with ISI \cite{Elia_MIMO_ITW,SabagKostinaHassibiMIMO}. The setups of \cite{Elia_MIMO_ITW,SabagKostinaHassibiMIMO} were unified in \cite{\cite{sabag22isit}} by studying the general dynamical system in \eqref{eq:intro_SS}, and formulate its capacity as a convex optimization. From a technical perspective, one of our main objectives is to extend this capacity result subject to an additional control cost constraint.




\textbf{LQG control:}
A fundamental problem in control is LQG control where the control cost in \eqref{eq:intro_LQR_cost} should be minimized subject to the state-space model in \eqref{eq:intro_SS}\footnote{Typically, the measurement $\vy_i$ is not directly affected by the control signal ($J=0$). However, in the CC problem, the measurement also plays a dual role of decoding information}. A policy (i.e., a controller) is a sequence of strictly-causal mappings from the measurements sequence to the control sequence. That is, at time $i$, the controller observes the noisy measurements of the state $\vy_1,\vy_2,\dots,\vy_{i-1}$, and chooses the control $\vx_i$. This is called a \emph{measurement-feedback} setting since the controller only has a noisy version of the state. The \emph{state-feedback} scenario is when the state is available directly to the controller, i.e., at time $i$, it has access to $\vs_i$. 

\textbf{The control-coding problem for discrete alphabets:} The CC problem when the state, input and output take values in discrete alphabets was studied in \cite{CC_discrete_IT_bambos,CC_LQG_IT_bambos,WuChenGunduz2025Actions}. In this case, the dynamical system can be represented by a non-linear transition kernel, and the cost may be taken as a general function. The general case of measurement feedback is challenging since it subsumes the feedback capacity problem where the capacity is known only in several instances \cite{Permuter06_trapdoor_submit,huleihel2023capacity,huleihel2025duality,Sabag_BEC,Sabag_BIBO_IT,Sabag_UB_IT}. For the state-feedback case, there are closed-form for the CC capacity expressions paralleling the feedback capacity of Markov channels that can be derived as a computable expression \cite{Chen05,POSTchannel,Eli_NOST_TIT,atayKostinaposet,WuChenGunduz2025Actions}. The main challenge when transitioning from state-feedback to measurement-feedback in the discrete case is similar to that when transitioning from Markov decision processes (MDPs) to partially-observable MDPs (POMDPs): the underlying control state becomes a belief which takes value in a simplex which is a continuous space.

\textbf{The Gaussian CC setup:}
The dynamical system in \eqref{eq:intro_SS} operates under \emph{measurement feedback}, since the controller and the observer have access to measurements of the system state. The special case in which the full state is available to both is called \emph{state feedback}; the capacity of such systems was studied in \cite{CC_discrete_IT_bambos,CC_LQG_IT_bambos} and further cast as a convex optimization in \cite{tanaka_allerton}.

For the measurement-feedback scenario studied here, a structured optimal policy for the $n$-letter optimization in \eqref{eq:n_letter_capacity} was proposed in \cite{charalambous2024siam,CharalambousLouka2025JSSC}. The policy exhibits a separation principle (see Section~\ref{subsec:policy}), but it does not yield computable capacity expressions as the policy is time-varying and therefore objective remains a multi-letter expression. To the best of our knowledge, there are no computable bounds or explicit expressions for the capacity of LQG systems in the measurement feedback scenario studied here.


\subsection{Contribution}
We derive an upper bound on the LQG system capacity as a finite-dimensional convex optimization. The upper bound recovers all special cases solved in the literature including feedback capacity of Gaussian channels with memory, LQG control and the LQG system capacity when there is state-feedback. This is the first computable converse bound on the LQG system capacity and it unifies the above special cases. The upper bound is accompanied with a sufficient condition for its tightness, expressed as a Riccati equation on the optimal decision variables.

For scalar dynamical systems, we establish the tightness of the upper bound. That is, their LQG system capacity is given by the convex optimization upper bound. This is established by showing that the optimal decision variables satisfy the sufficient condition via a through study of the KKT conditions for the optimization of the upper bound.  

For vector dynamical systems, it is observed that the sufficient condition is approximately satisfied in all tested examples. Motivated by this observation, we extract policies from the upper bound and evaluate their performance in terms of control cost and communication rates. The extracted policies yield lower bounds on the LQG system capacity that appear to be tight in all tested examples.


\subsection{Paper Structure}
Section \ref{sec:setup} introduces the notation and provides preliminaries on Kalman filtering and LQG control. Section \ref{sec:main} presents the main results. Section \ref{sec:examples} shows numerical examples. Section \ref{sec:derivation} outlines the derivation of the main results, and Section \ref{sec:proofs} includes proofs for technical lemmas. The paper is concluded in Section \ref{sec:conclusion}.

\section{Preliminaries}\label{sec:setup}
This section presents the notation and reviews preliminaries on the Kalman filter (KF) and linear–quadratic–Gaussian (LQG).
\subsection{Notation}
Random variables and random vectors are denoted by bold lowercase letters such as $\mathbf{x}$, deterministic matrices are denoted by uppercase letters (e.g. $F$), and matrix transposes by $F^T$. Superscripts excluding $T$, indicate sequences,  $\vx^i = (\vx_1, \vx_2, \dots,\vx_i)$, and $A^\dagger$ denotes the Moore–Penrose pseudo inverse of $A$. The spectral radius of a square matrix $M$, denoted $\rho(M)$, is defined as $\rho(M) = \max\{\,|\lambda| : \lambda \in \sigma(M)\,\}$.

\subsection{Reduction via Kalman Filter}
The state-space model in \eqref{eq:intro_SS} has a measurement-feedback and, therefore, the controller-encoder does not have access to the actual state $\vs_i$. In this section, we reduce the general CC Gaussian problem to that of a state-space model where the state is available to the controller-encoder. The new state is the MMSE estimate of the state, and the reduction is standard and is based on the Kalman filter (KF). 

Consider the MMSE of the state given the controller-encoder information
\begin{align}\label{eq:preliminary_vhs}
    \hat{\vs}_i &\triangleq \E[\vs_i|\vx^{i-1},\vy^{i-1}].
\end{align}
By KF, the estimate can be recursively expressed as
\begin{align}\label{eq:encoder_estimate}
\hat{\vs}_{i+1} &= F \hat{\vs}_i +G\vx_i +K_{p,i}(\vy_i-J\vx_i-H\hat{\vs}_i)
\end{align}
with $\hat{\vs}_1=0$. The matrix $K_{p,i} = (F\Sigma_iH^T+L)\Psi_i^{-1}$ is the \emph{Kalman gain} and $ \Psi_i =H\Sigma_i H^T + V$ is the innovation variance. Both terms are a function of the estimation error covariance, $\Sigma_i\triangleq\mathbf{cov}(\vs_i-\hat{\vs}_i)$, which can be computed by the Riccati recursion
\begin{align}\label{eq:cov_encoder_err}
    \Sigma_{i+1}=F\Sigma_iF^T+W-K_{p,i}\Psi_i K_{p,i}^T,
    \end{align}
initialized at $\Sigma_1$. The recursion \eqref{eq:cov_encoder_err} is a Riccati recursion, and we assume regularity conditions to guarantee convergence of the recursion to the stabilizing solution of the Riccati equation
\begin{align}\label{eq:riccati_sigma}
    \Sigma = F\Sigma F^T +W - K_p\Psi K_p^T
\end{align}
where $K_p = (F\Sigma H^T +L)\Psi^{-1}$ and $\Psi = H\Sigma H^T + V$. The regularity conditions and the stabilizing solution to the Riccati equation are standard in Control and are presented for completeness in Appendix \ref{app:sub_sigma}. The convergence rate of the Riccati recursion to the stabilizing solution is geometric~\cite{kailath_booklinear}. Therefore, from this point onward, we utilize the steady-state parameters $K_p$ and $\Psi$, evaluated at the stabilizing solution $\Sigma$, and omit their time indices. 

We can write a new state-space model that represents the system dynamics 
\begin{align}\label{eq:encoder_model}
    \hat{\vs}_{i+1} &=F \hat{\vs}_i +G\vx_i +K_{p}\mathbf{e}_i\nonumber\\ 
    \vy_i &=H \hat{\vs}_i +J\vx_i +\mathbf{e}_i,
\end{align}
where the main advantage is that the new state $\hat{\vs}_i$ is available to the controller-encoder. The random noise inserted into this state-space model is 
\begin{align}\label{eq:primary_err}
    \mathbf{e}_i &\triangleq \vy_i - J\vx_i - H\hat{\vs}_i,
\end{align}
known as the innovation process in KF, which is an i.i.d. random process and its distribution is given by $\mathbf{e}_i\sim \mathcal{N}(0, \Psi)$.

The LQR cost in \eqref{eq:intro_LQR_cost} can also be expressed with respect to the new state 
\begin{align}\label{eq:cost_by_hat_s}
   \mathcal{J}_n&=\frac1{n}\left[ \E[\hat{\vs}_{n+1}^T Q \hat{\vs}_{n+1}] + \sum_{i=1}^n (\E[\hat{\vs}_i^T Q \hat{\vs}_i] + \E[\vx_i^TR\vx_i]) \right]  + \frac{n+1}{n}\Tr(\Sigma Q),
\end{align}
where we used the orthogonality principle in MMSE estimation as $\E[\vs_i^T Q \vs_i] = \E[(\vs_i - \hat{\vs}_i)^T Q (\vs_i - \hat{\vs}_i) ] + \E[\hat{\vs}_i^T Q \hat{\vs}_i]$ and $\E[(\vs_i - \hat{\vs}_i)^T Q (\vs_i - \hat{\vs}_i) ]= \Tr(\Sigma Q)$.

\subsection{LQG optimal control}\label{subsec:lqr}
The LQG control problem is a special case of the CC Gaussian problem considered here. The objective in LQG control is to minimize the control cost \eqref{eq:intro_LQR_cost} subject to the state-space model in \eqref{eq:intro_SS}. Optimal control theory provides with the optimal policy and the minimal control cost. The optimal policy is given by 
\begin{align}\label{eq:lqg_policy}
    \vx^{\mathtt{LQG}}_i= -K_{\mathtt{LQR},i}\hat{\vs}_i,
\end{align} 
where $\hat{\vs}_i$ is defined in \eqref{eq:encoder_estimate} and $\KLQRi$ is a matrix that depends on the system parameters. 
 
A key distinction between LQG control and our setting is the information available to the controller and the observer. In LQG control, the optimal control signal \eqref{eq:lqg_policy} is a deterministic function of the measurements; therefore, an observer who knows the measurements can reconstruct the control signal. In our setup, the control signal is also used for communication and thus depends on a message that is only known to the controller. This induces a degraded information structure: both the controller and the observer have access to the measurements, but only the controller can reconstruct the optimal control signal.

The controller gain $\KLQRi$ is computed via the backward Riccati recursion
\begin{align}\label{eq:control_recursion}
    E_i &= F^TE_{i+1}F+Q - \KLQRi^T (R+G^TE_{i+1}G) \KLQRi \ \ i=1,\dots,n \nonumber\\
    \KLQRi &= (R+G^TE_{i+1}G)^{-1}G^TE_{i+1}F 
\end{align}
with the terminal condition $E_{n+1} = Q$.
Under standard regularity conditions, detailed in the Appendix \ref{app:sec_lqg}, the recursion \eqref{eq:control_recursion} converges to the stabilizing solution of the Riccati equation
\begin{align}\label{eq:lqr_Riccati_equation}
    E &= F^TE F + Q - K_{\mathtt{LQR}}^T \Psi_{\mathtt{LQR}} K_{\mathtt{LQR}}\nonumber\\
    K_{\mathtt{LQR}} &= \Psi_{\mathtt{LQR}}^{-1}G^TEF \nonumber\\
    \Psi_{\mathtt{LQR}} &= R+G^TEG.
\end{align}
The minimal LQR cost \eqref{eq:intro_LQR_cost} for the LQG problem  in the infinite-horizon regime is $\mathcal J^*\triangleq \lim_{n\to\infty} \min_{\pi_n} \mathcal J _n(\pi_n) = \Tr(K_p \Psi K_P^T E) + \Tr(\Sigma Q)$. Clearly, this value is a lower bound on the control budget required to achieve positive capacity in our problem.

\section{Main Results}\label{sec:main}
In this section we present our main results. We first present a general upper bound on the capacity of LQG control systems, and then provide a sufficient condition for its tightness. Using the upper bound and the sufficient condition, we establish a capacity in the scalar case. Lastly, we show that all special cases solved in the literature coincide with our upper bound.
\begin{theorem}[Upper bound on the capacity of LQG systems]\label{theorem:upper_bound} The capacity of the LQG system is upper bounded by the convex optimization
\begin{align}\label{eq:main_UB}
    C_{\mathtt{LQG}}(p) \leq &\max_{\Pi\in\mathbb{R}^{m\times m},\Gamma \in\mathbb{R}^{m \times k}, \hat{\Sigma} \in\mathbb{R}^{k\times k}}  \frac{1}{2}\log \left(\frac{\det (\Psi_{Y})}{ \det (\Psi)}\right)\\
        s.t. \quad &   \Tr( \hat{\Sigma}K_{\mathtt{LQR}}^T \Psi_{\mathtt{LQR}}K_{\mathtt{LQR}}) + \Tr(\Pi  \Psi_{\mathtt{LQR}}) + 2 \Tr(\Gamma K_{\mathtt{LQR}}^T \Psi_{\mathtt{LQR}}) + \Tr(K_p\Psi K_p^TE) + \Tr (\Sigma Q) \leq p \nonumber\\
        &\begin{pmatrix}
        \Pi & \Gamma\\  \Gamma^T & \hat{\Sigma}
    \end{pmatrix} \succeq 0 \nonumber \\&
    \begin{pmatrix}
     F\hat{\Sigma} F^T+F\Gamma G^T+ G\Gamma F^T + G\Pi G^T+  K_{p}\Psi K_{p}^T - \hat{\Sigma}& K_{Y}\Psi_{Y}\\ 
      \Psi_{Y}K_{Y}^T & \Psi_{Y}
    \end{pmatrix} \succeq 0\nonumber\\
    &\Psi_{Y} = J\Pi J^T +H\hat{\Sigma} H^T+H\Gamma^TJ^T+ J\Gamma H^T+ \Psi \nonumber \\ 
    &  K_{Y}=  (F\Gamma J^T + F \hat{\Sigma} H^T + G\Pi J^T + G \Gamma H^T +K_{p}\Psi) \Psi_{Y}^{-1},
    \end{align}
where $(K_{\mathtt{LQR}}, \Psi_{\mathtt{LQR}}, E)$ are the LQR constants in  \eqref{eq:lqr_Riccati_equation}, and $(K_p, \Psi, \Sigma)$ are the KF constants in \eqref{eq:riccati_sigma}.
\end{theorem}
The upper bound is expressed as a convex optimization, and thus can be computed efficiently using standard software. The directed information in \eqref{eq:n_letter_capacity} is translated to the objective function $ \frac{1}{2}\log \left(\frac{\det (\Psi_{Y})}{\det (\Psi)}\right) $, and the trace constraint corresponds to the LQR control cost, bounded by the control budget $p$. Note that the last two terms in the trace constraint are the constants $ \Tr(K_p\Psi K_p^TE) + \Tr (\Sigma Q)$, which correspond to the optimal LQG control cost; hence, if $p$ is lower than their sum, communication is not feasible since the minimal control cost is greater than the control budget. The linear matrix inequalities (LMIs) in \eqref{eq:main_UB} reflect structural properties of the optimal policy, and are discussed next.

\subsection{The optimal policy and LMIs}\label{subsec:policy}
It is shown below (Lemma \ref{lemma:Gaussian_linear_policy} and \ref{lemma:minimizig_Gamma_2}) that an optimal policy is given by 
\begin{align}\label{eq:main_policy}
    \vx_i&= - K_{\mathtt{LQR},i} \dhat{\vs}_i + \Gamma_i \hat{\Sigma}^\dagger_i (\hat{\vs}_i - \dhat{\vs}_i) + \mathbf{m}_i,
\end{align}
where $\hat{\vs}_i$ is defined in \eqref{eq:encoder_estimate}, $\dhat{\vs}_i \triangleq \E[\hat{\vs}_i|\vy^{i-1}]$ is the state estimate at the observer-decoder, $\hat{\Sigma}_i = \mathbf{cov}(\hat{\vs}_i - \dhat{\vs}_i)$ is the corresponding error covariance, and $\mathbf{m}_i$ is an i.i.d. Gaussian process. A similar policy was derived in \cite{CC_Bambos_ISIT24} where the main difference is the normalization factor $\hat \Sigma^\dagger$ in the second term of the policy, required for the convex optimization formulation. The upper bound corresponds to a time-invariant policy where the time indices are removed from $\Gamma_i, \hat{\Sigma}_i,K_{\mathtt{LQR},i}$ and the additive term distribution is $\mathbf{m}_i\sim \mathcal N(0, \Pi - \Gamma \hat{\Sigma}^\dagger \Gamma^T)$.

Few remarks on the policy are in order:
\begin{enumerate}
    \item The first term $- K_{\mathtt{LQR}} \dhat{\vs}_i$ is a pure control element. It is similar to the optimal LQG policy $- K_{\mathtt{LQR}} \hat{\vs}_i$ in \eqref{eq:lqg_policy}, but the state estimate at the controller $\hat{\vs}_i$ is replaced with the state estimate at the observer~$\hat{\hat{\vs}}_i$. As such, it does not carry information about the message since the observer has access to $\dhat{\vs}_i$.
    \item The remaining terms $\Gamma \hat{\Sigma}^\dagger_i (\hat{\vs}_i - \dhat{\vs}_i)+ \mathbf{m}_i$ are not known to the observer and therefore may carry information (about the message) from the controller to the observer. The first term is a linear combination of the estimation error, while the second term $\mathbf{m}_i$ includes \emph{information} in the form of a random vector that is independent of the past occurrences $(\vx^{i-1},\vy^{i-1})$. We remark that these two elements also appear in solutions to the communication problem when state regularization is not required, e.g., \cite{Kim10_Feedback_capacity_stationary_Gaussian}. Lastly, note that if state-feedback is available, we have $\hat{\vs}_i - \dhat{\vs}_i=0$ so that only $\mathbf{m}_i$ carries the information, see \cite{tanaka_allerton,bambos_IT_CC_FI}. 
\end{enumerate}

Using the optimal policy, the LMIs in \eqref{eq:main_UB} can be interpreted; The matrix $\Pi$ serves as an upper bound in the PSD order for the covariance of the information terms, i.e., $\mathbf{cov}(\Gamma_i \hat{\Sigma}^\dagger_i (\hat{\vs}_i - \dhat{\vs}_i) + \mathbf{m}_i) = \Gamma \hat{\Sigma}^\dagger\Gamma^T+ M$. 
By the Schur complement, the first LMI implies that $\Pi \succeq \Gamma \hat{\Sigma}^\dagger\Gamma^T$ once the decision variable $M\succeq0$ is eliminated. The Schur complement of the second LMI in \eqref{eq:main_UB} can be written as 
\begin{align}\label{eq:Schur_2LMI}
    F\hat{\Sigma} F^T +F\Gamma G^T+ G\Gamma F^T + G\Pi G^T+  K_{p}\Psi K_{p}^T - \hat{\Sigma} -  K_{Y} \Psi_{Y}K_{Y}^T.
\end{align}
The LMI positivity implies the Schur \eqref{eq:Schur_2LMI} is PSD, and it corresponds to a relaxation of the Riccati equation to the Riccati inequality
\begin{align}\label{eq:Riccati_ineq_discussion}
    \hat{\Sigma} \preceq F\hat{\Sigma} F^T+F\Gamma G^T+ G\Gamma F^T + G\Pi G^T+  K_{p}\Psi K_{p}^T -  K_{Y} \Psi_{Y}K_{Y}^T.
\end{align}
The relaxation of the Riccati equation to a Riccati inequality is required for the formulation of the upper bound as convex optimization. Indeed, showing that this relaxation is achieved with equality under the optimization in Theorem \ref{theorem:upper_bound} is the remaining gap for establishing the capacity of LQG systems as the convex optimization in \eqref{eq:main_UB}. This is formalized in the next theorem as a sufficient condition for the tightness of the upper bound.
\begin{theorem}[Sufficient condition for the tightness of the upper bound]\label{theorem:achievabilty}
    The LQG system capacity is equal to the upper bound in \eqref{eq:main_UB} if the optimal variables, $(\Pi^*,\Gamma^*,\hat{\Sigma}^*)$, satisfy the Riccati equation 
    \begin{align}\label{eq:main_ricc_eq_sufficient}
     \hat{\Sigma}^*&= F\hat{\Sigma}^* F^T+F\Gamma^{*T}G^T+ G\Gamma^* F^T + G\Pi^* G^T+  K_{p}\Psi K_{p}^T - K_{Y}^*\Psi_{Y}^*K_{Y}^{*T},
    \end{align}
    where $K_{Y}^*,\Psi_{Y}^*$ are evaluated at the optimal variables, and the Kalman filter regularity conditions are met  (see Appendix \ref{app:sub_hat_sigma}): (i) The pair $(F + G \Gamma^* \hat{\Sigma}^{*\dagger}, H+J\Gamma^* \hat{\Sigma}^{*\dagger})$ is detectable (ii) The pair $(F^s, G^s W^{s})$ is stabilizable, where 
\begin{align}\label{eq:independent_riccati}
    F^s &= F+G\Gamma^* \hat{\Sigma}^{*\dagger}-(GM^*J^T+K_p\Psi)(JM^*J^T + \Psi)^{-1}(H+J\Gamma^* \hat{\Sigma}^{*\dagger})^T\\
    W^s &= \begin{pmatrix}
  M^* & 0\\ 
  0 & \Psi
\end{pmatrix} - 
\begin{pmatrix}
  M^*J^T \\ 
  \Psi
\end{pmatrix} (JM^*J^T+\Psi)^{-1}
\begin{pmatrix}
  JM^* &\Psi
\end{pmatrix} \nonumber\\
    G^s&=\begin{pmatrix} G &K_p \end{pmatrix}\nonumber
\end{align}
with $M^* \triangleq  \Pi^*-\Gamma^* \hat{\Sigma}^{*\dagger}\Gamma^{*T}$.
\end{theorem}
The main idea in proving Theorem \ref{theorem:achievabilty} is to use a general lower bound on the LQG capacity for time-invariant policies parameterized by $(\Gamma,M)$ and induces $\hat{\Sigma}_{LB}$. We extract a specific pair $(\Gamma,M)$ from the upper-bound optimization, and using the sufficient Riccati condition, show that the policy-induced-covariance $\hat{\Sigma}_{LB}$ is equal to the upper-bound solution $\hat{\Sigma}^*$. 

In Section \ref{sec:examples}, we demonstrate via numerical simulations the upper bound tightness in the simulated dynamical systems. Moreover, the upper bound yields policy parameters $\Gamma, M$ that numerically achieve a lower bound indistinguishable from the upper bound. While we could not prove tightness for the general case, we were able to prove it for the case of scalar dynamical systems, i.e., all variables are scalars. 
\subsection{The capacity of scalar dynamical systems}
This is our main result regarding scalar dynamical systems. 
\begin{theorem}[The capacity of scalar LQG systems]\label{theorem:scalar}
    Assume that $H \neq K_{\mathtt{LQR}}J$, $G \neq K_p J$, and that the capacity is positive. Then, the capacity of scalar LQG systems is equal to
    \begin{align}\label{eq:capacity_scalar}
    C_{\mathtt{LQG}}(p) &= \max_{\Pi,\Gamma, \hat{\Sigma}}  \frac{1}{2}\log \left(\frac{\Psi_{Y}}{ \Psi}\right)\nonumber\\
        & s.t. \quad (\hat{\Sigma}K_{\mathtt{LQR}}^2 +\Pi +2 \Gamma K_{\mathtt{LQR}}) \Psi_{\mathtt{LQR}} + \Tr(K_p^2\Psi E) + \Tr (\Sigma Q) \leq p \nonumber\\
    &\Psi_{Y} = J^2\Pi + H^2 \hat{\Sigma} + 2 H\Gamma J+ \Psi \nonumber \\ 
    &  K_{Y}=  (F\Gamma J + F \hat{\Sigma} H + G\Pi J^T + G \Gamma H +K_{p}\Psi) \Psi_{Y}^{-1} \nonumber \\
    &\begin{pmatrix}
        \Pi & \Gamma\\  \Gamma^T & \hat{\Sigma}
    \end{pmatrix} \succeq 0 \nonumber \\&
    \begin{pmatrix}
     F^2\hat{\Sigma} + 2 F\Gamma G + G^2\Pi +  K^2_{p}\Psi  - \hat{\Sigma}& K_{Y}\Psi_{Y}\\ 
      \Psi_{Y}K_{Y}^T & \Psi_{Y}
    \end{pmatrix} \succeq 0.
    \end{align}
\end{theorem}
\begin{remark}
The assumption $G \neq K_pJ$ in Theorem \ref{theorem:scalar} implies that the dynamical system has a measurement feedback, and not a state-feedback. This is implied by the state evolution of the controller's estimate $\hat \vs_i$ in \eqref{eq:encoder_estimate}
\begin{align*}
\hat{\vs}_{i+1} =F \hat{\vs}_i +G\vx_i +K_{p}(\vy_i-J\vx_i-H\hat{\vs}_i).
\end{align*}
If $G = K_pJ$, the estimate update is only a function of $\vy_i$, and thus can be reproduced at the observer. This means that the estimation error covariance is zero, $\hat\Sigma=0$, and the solution for the CC problem can be found via the solution to the state-feedback scenario in \cite{tanaka_allerton}.
\end{remark}

The main idea to prove Theorem \ref{theorem:scalar} is to verify the sufficient condition in Theorem \ref{theorem:achievabilty}. The main challenge follows from the role of $\hat{\Sigma}$ in the optimization of \eqref{eq:main_UB}. The objective function $\Psi_Y$ is increasing with $\hat{\Sigma}$, so we aim to find the maximal $\hat{\Sigma}$ subject to the optimization constraints. However, both the Riccati LMI constraint (via \eqref{eq:Riccati_ineq_discussion}) and the cost constraint provide with upper bounds on $\hat \Sigma$ or the trace of it. To show the sufficient condition, one need to establish that the upper bound on $\hat \Sigma$ implied by the Riccati inequality is achieved with equality, and in other words, this constraint is active at the optimum. This is shown by analyzing the KKT conditions for the optimization of the upper bound. 

\subsection{The upper bound in known special cases}
Following the above discussion on the role of the constraints in \eqref{eq:main_UB} on the decision variable $\hat{\Sigma}$, we can show that the upper bound is tight for all special cases solved in the literature. 
\begin{enumerate}
    \item \textbf{Feedback capacity of Gaussian channels:} In this special case, the control cost is specialized to a power constraint on the system inputs by setting \( Q = 0 \). This implies \( E = 0 \), \( K_{\mathtt{LQR}} = 0 \), and \( \Psi_{\mathtt{LQR}} = R \). The control cost constraint is then \( \Tr(\Pi R) \leq p \), and the upper bound recovers the feedback capacity derived in~\cite{sabag22isit}.

    \item \textbf{The CC problem with state-feedback:} Here, the state is directly available to the observer-decoder, so $L = V = W$ and \( \Sigma = \hat{\Sigma} = 0 \). By the first LMI constraint, we have \( \Gamma = 0 \), and the objective function simplifies to \( \Psi_Y = J\Pi J^T + \Psi \). The control cost is given by \( \Tr(M \Psi_{\mathtt{LQR}}) + \Tr(WE) \leq p \), all align with the optimization for the capacity of state-feedback systems in \cite{tanaka_allerton}.
        \item \textbf{LQG control:} In this case, no information is transmitted, and the policy is \( \vx_i = - K_{\mathtt{LQR}} \dhat{\vs}_i \) with \( \Gamma = 0 \) and \( M = 0 \). Consequently, \( \Pi = 0 \), and since the input is deterministic with respect to the measurements, we have \( \hat{\Sigma} = 0 \) and \( \Psi_Y = \Psi \). The control constraint simplifies to the minimal LQG cost \( \Tr(K_p\Psi K_p^TE) + \Tr (\Sigma Q) \leq p \), as expected.

\end{enumerate}

\section{Examples}\label{sec:examples}
This section presents several examples: in the scalar case we demonstrate properties of the LQG system capacity, and in the vector case we illustrate that the upper bound is tight.
\subsection{Scalar systems}\label{subsec:example_scalar}
In the case of scalar dynamical systems, the LQG system capacity is given in Theorem \ref{theorem:achievabilty} as a convex optimization. We illustrate several phenomena related to the policy structure discussed in Section \ref{subsec:policy} via the dynamical system 
\begin{align}\label{eq:ex_scalar}
    \vs_{i+1}&=0.5 \vs_i +  \vx_i + \mathbf{w}_i\nonumber\\
    \vy_i &= \vs_i + \vx_i + \mathbf{v}_i,
\end{align}
where $\mathbf{w}_i\sim\mathcal{N}(0,1), \mathbf{v}_i\sim\mathcal{N}(0,1)$, their covariance is $L=0$, and the control cost parameters are $Q=1, R=1$.

We plot the LQG system capacity in Fig. \ref{fig:M_zero} as a function of the control budget. For control budgets smaller than the LQG minimal cost, $\mathcal J^*$, the capacity is zero. The blue curve in Fig. \ref{fig:M_zero} corresponds to the variance of the i.i.d. element $\mathbf{m}_i$ in \eqref{eq:main_policy}, equal to $\Pi - \frac{\Gamma^2}{\hat \Sigma}$. It can be noted that this policy term does not exist as its variance is zero even as $p$ increases. This observation has direct implications on the construction of explicit coding schemes to achieve the capacity of this LQG system capacity. In particular, an optimal scheme is to quantize the message on the unit interval and send this quantization at the first time instance. At the following times, the controller-encoder follows the optimal policy 
\begin{align}
   \vx_i&=   - K_{\mathtt{LQR},i} \dhat{\vs}_i + \Gamma \hat{\Sigma}^\dagger (\hat{\vs}_i - \dhat{\vs}_i) \quad \quad i \ge 2
\end{align}
where $\Gamma, \hat{\Sigma}$ are the optimal parameters from the upper bound. The decoder is constructed by the MMSE estimator of the $z_0:= H\vs_0 + \vv_i$ based on the measurements $\vy^n$ with explicit construction in \cite{SabagKostinaHassibiMIMO}. An analysis of the rate achieved by such scheme can also be found in \cite{SabagKostinaHassibiMIMO} and shows that it achieves the capacity. The main idea is that the joint distribution induced by the suggested scheme converges to the optimal joint distribution in Theorem \ref{theorem:scalar}. This scheme falls resembels the Schalkwijk-Kailath scheme \cite{SchalkwijkKailath66_feedback_scheme}, which falls under the mechanism of posterior matching \cite{shayevitz_posterior_mathcing}.  However, the system here is not memoryless and rather than transmitting the estimation error of the message, the controller transmits the estimation error of the state to implicitly refine the observer knowledge on the message. To our knowledge, the idea of transmitting the states' innovation to convey information on the message was introduced in \cite{Kim10_Feedback_capacity_stationary_Gaussian}.
\begin{figure}[H]
        \centering
        \includegraphics[width=0.5\textwidth, trim=3cm 5cm 3cm 5cm, clip]{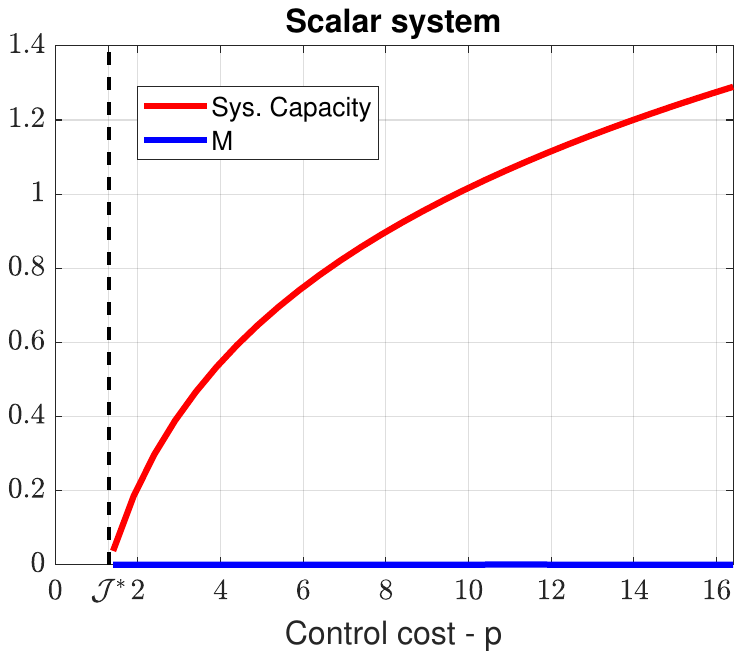}
        \caption{The LQG system capacity for the scalar system in \eqref{eq:ex_scalar}. We also plot the variance  $V(\mathbf{m}_i) = \Pi - \frac{\Gamma^2}{\hat \Sigma}$ to illustrate that $\mathbf{m}_i =0$ in this example. The dashed vertical line corresponds to the minimal LQG control cost. .}
        \label{fig:M_zero}
    \end{figure}

It may be tempting to conjecture that the suggested scheme is optimal for all LQG systems, but even a simple modification of the system parameters prove the opposite. In Fig. \ref{fig:M_from_G}, we plot the LQG system capacity for the dynamical system studied in Fig. \ref{fig:M_zero}, but as a function of $G$ when the control cost is fixed to $p=5$. We observe that there are regimes in which $ \Pi - \Gamma \hat{\Sigma}^\dagger \Gamma^T >0$ so that the i.i.d. term in the policy $\mathbf{m}_i \neq0$. To the best of our knowledge, there is no explicit coding scheme when the optimal joint distribution induces $\mathbf{m}_i\neq0$, see also \cite{SabagKostinaHassibiMIMO,comments_kim}. While in the case of feedback capacity, it is unknown whether schemes such as \cite{SabagKostinaHassibiMIMO,Kim10_Feedback_capacity_stationary_Gaussian} achieve the feedback capacity \cite{comments_kim}, it is demonstrated here that such constructions are insufficient to achieve the LQG system capacity.

\begin{figure}
        \centering
        \includegraphics[width=0.5\textwidth, trim=3cm 5cm 3cm 5cm, clip]{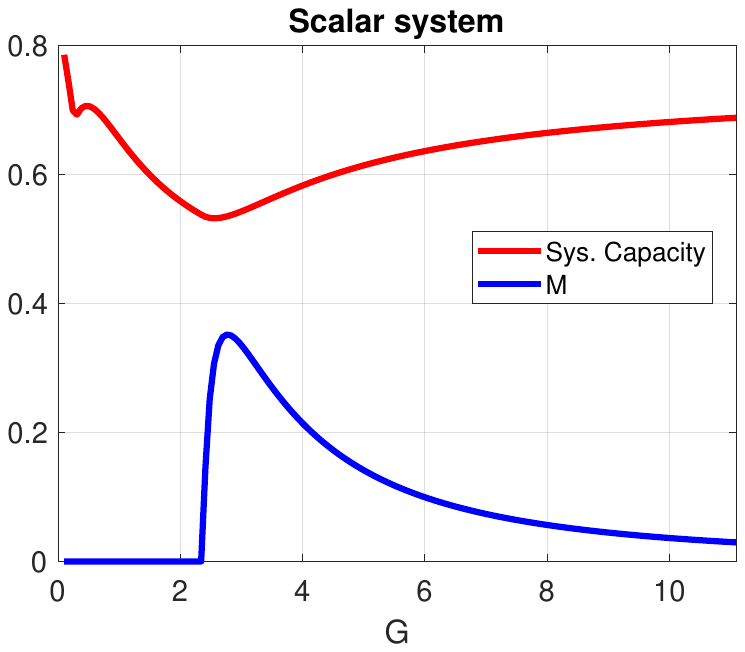}
        \caption{Illustration of the LQG system capacity and $M = \Pi - \Gamma \hat{\Sigma}^\dagger \Gamma^T$ for a fixed control budget $p=5$ and varying values of the input-to-state gain $G$.}
        \label{fig:M_from_G}
    \end{figure}
    
\subsection{Vector systems}\label{subsec:example}
We simulate an unstable vector dynamical system 
$\vs_i \in \mathbb{R}^3$ and scalar inputs and outputs. The paramaters of the systems are
\begin{align}\label{eq:vector_system}
    \vs_{i+1}&=\begin{pmatrix}
        1.2 & 0 &0 \\ 0 &0.7 &0 \\0&0&0.5
    \end{pmatrix}
     \vs_i + \begin{pmatrix}
       2 \\ 1 \\12
    \end{pmatrix} \vx_i + \mathbf{w}_i\nonumber\\
    \vy_i &= \begin{pmatrix}
        10 & 2 &1 
    \end{pmatrix} \vs_i + \vx_i + \mathbf{v}_i,
\end{align}
where the control cost parameters are $Q=I_3, R=1$, and the distribution of the disturbance and the measurement noise is characterized by $\mathbf{w}_i \sim N(0,I_3), \mathbf{v}_i \sim N(0, 4)$ and $L= \mathbb{E}[\vw_i \vv_i]= \begin{bmatrix}
    0&0&0
\end{bmatrix}^T$. 

The LQG system capacity for this scenario is not given in a closed form since the state is a vector. Fig. \ref{fig:like_AR} presents upper and lower bounds on the capacity as a function of the control budget. The upper bound is computed using the convex optimization in Theorem \ref{theorem:upper_bound}. For the lower bound, we extract the policy parameters $\overline{\Gamma} = \Gamma \hat{\Sigma}^\dagger, M = \Pi - \Gamma \hat{\Sigma}^\dagger \Gamma^T$ from the upper bound optimization and then evaluate the performance of this policy. This is done by computing $\hat{\Sigma}$ as the solution of a Riccati equation 
\begin{align}\label{eq:hat_sigma_riccati}
    \hat{\Sigma}= (F+G \overline{\Gamma})\hat{\Sigma}(F+G \overline{\Gamma})^T + GMG^T + K_p\Psi K_p^T -\overline{K}_{Y}\overline{\Psi}_{Y}\overline{K}_{Y}^T
\end{align}
where $\overline{K}_{Y}=  ((F+G \overline{\Gamma} )\hat{\Sigma}(H+J\overline{\Gamma})^T + GMJ^T +K_p\Psi)\overline{\Psi}_Y^{-1}$ and $\overline{\Psi}_{Y} = (H+J\overline{\Gamma})\hat{\Sigma} (H+J\overline{\Gamma})^T + J M J^T +  \Psi.$

This is used to compute the objective function and the control cost consumed by this strategy. As can be seen in Fig. \ref{fig:like_AR}, the lower and upper bounds seem to coincide. Moreover, in all simulated examples it is observed that the lower and upper bound exhibit similar tightness. We also remark that for large dimension, we observe large numerical errors, possibly due to the inverse $\hat \Sigma$.


    \begin{figure}
        \centering
        \includegraphics[width=0.5\textwidth, trim=3cm 5cm 3cm 5cm, clip]{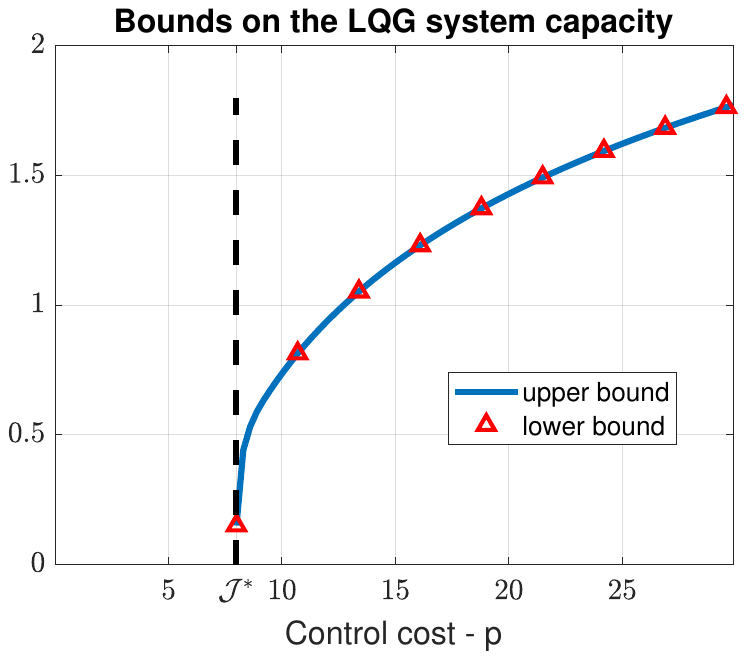}
        \caption{Bounds on the LQG system capacity for the vector dynamical system in \eqref{eq:vector_system}. The dashed vertical line indicates the minimum control budget to achieve a positive capacity. }
        \label{fig:like_AR}
    \end{figure}

\section{Derivation of the Main results}\label{sec:derivation}

In this section, we first present technical lemmas that are necessary to present the proof of the main results. In Section \ref{subsec:theorems_proof} we prove the main results based on the technical lemmas. The proof of the technical lemmas appear in Section \ref{sec:proofs}.
Our starting point is the $n$-letter optimization problem in \eqref{eq:n_letter_capacity}. We also present the $n$-letter  joint distribution of the underlying objective:
\begin{align}
    P(\vs_1) \prod_{i=1}^n P(\vs_{i+1},\vy_i|\vs_i,\vx_i)P(\vx_i|\vx^{i-1},\vy^{i-1}).
\end{align}

\begin{lemma}\label{lemma:Gaussian_linear_policy}
An optimal policy for the $n$-letter optimization problem in \eqref{eq:n_letter_capacity}, it is sufficient to optimize over inputs of the form
 \begin{align}\label{eq:linear_policy}
     \vx_i= \Gi \sd \Err{i} + \Gti \dhat{\vs}_i + \mathbf{m}_i, \ \ i=1,\dots,n, 
 \end{align}
where $\hat\vs_i$ is defined in \eqref{eq:encoder_estimate}, $\dhat{\vs}_i$ is the observer estimate of the state defined in \eqref{eq:encoder_model}, $\Gi \in \mathbb{R}^{m\times k}$, $\Gti \in \mathbb{R}^{m\times k}$, $\hat{\Sigma}_i \triangleq \mathbf{cov} \Err{i}$ and $\mathbf{m}_i \sim N(0,M_i)$ is independent of $(\vx^{i-1},\vy^{i-1})$. The optimization variables are the collection of matrices $\{\Gamma_{1,i},\Gamma_{2,i},M_i\}_{i=1}^n$, and are subject to the orthogonality constraint
\begin{align}\label{eq:orthogonal}
    \Gi(I-\hat{\Sigma}_i \sd) =0 \ \  \text{for} \ \ i=1,\dots,n, 
\end{align}
and the control cost constraint in \eqref{eq:cost_by_hat_s}. 
\end{lemma}

Lemma \ref{lemma:Gaussian_linear_policy} simplifies the optimization in \eqref{eq:n_letter_capacity} by showing that an optimal policy is linear, Gaussian, and has a simple dependence on past occurrences, i.e., the system input $\vx_i$ only depends on $\hat{\vs}_i$ and $\dhat{\vs}_i$ while in principle it can depend on $(\vx^{i-1},\vy^{i-1})$. These, as a whole, imply that the optimization domain is the matrices' sequence $\{\Gi, \Gti, M_i \succeq 0\}_{i=1}^n$ subject to the orthogonality and the control constraints.

The existence of an optimal Gaussian policy for the control-coding problem is not direct given the non-aligned objectives of communication and control. However, the underlying principle is that in both problems, second-order moments play a significant role. In Information theory, for a fixed-variance, the entropy is maximized at Gaussian distributions, and similarly in control the LQR objective involves second-order moments and thus linear Gaussian policies are also optimal. 
We remark that a similar policy was derived in \cite{CC_Bambos_ISIT24} with the exception that we introduce the orthogonality constraint in \eqref{eq:orthogonal}, required for the convex optimization formulation, and thus we use a constructive proof. 

A consequence of Lemma~\ref{lemma:Gaussian_linear_policy} is that the LQG system in \eqref{eq:encoder_model} can be written as the state-space model
\begin{align}\label{eq:observer_model}
        \hat{\vs}_{i+1} &= (F+G\Gi \sd) \Err{i} + (F+G \Gti) \dhat{\vs}_i + G\mathbf{m}_i + K_p \mathbf{e}_i \nonumber\\
        \vy_i &= (H+J\Gi \sd) \Err{i} + (H+ J \Gti) \dhat{\vs}_i + J \mathbf{m}_i  + \mathbf{e}_i.
\end{align}
We now utilize \eqref{eq:observer_model} to express the optimal estimator $\dhat{\vs}_{i} = \E[\hat{\vs}_i|\vy^{i-1}]$ as a simple linear and recursive function of the measurements (channel outputs). While the state-space resembles a Kalman filtering model, the standard Kalman filtering equations cannot be applied directly. In particular, note that the latent state in the first equation of \eqref{eq:observer_model} is $\hat{\vs}_{i+1}$ and evolves as a function of the previous state $\hat{\vs}_{i}$ but also as a function of the estimate $\dhat{\vs}_{i}$. The dependence on $\hat{\vs}_{i}$ is standard, but the term $\dhat{\vs}_i$ is a strictly-causal function of the measurements (channel outputs) $\vy^{i-1}$. As we will see in the next lemma, optimal estimate has the standard structure of the Kalman-filter, but the Kalman gain and the constants are different. The estimator is linear and is given by
\begin{align}\label{eq:dhat_discussion}
    \dhat{\vs}_{i+1}&=\E[\hat{\vs}_i|\vy^{i-1}]+K_{Y,i}(\vy_i -\E[\vy_i|\vy^{i-1}])\nonumber\\
    &= (F + G\Gti) \dhat{\vs}_i + K_{Y,i} \psi_i
\end{align}
where the innovation is $\psi_i \triangleq \vy_i - \E[\vy_i|\vy^{i-1}] = \vy_i-(H + J \Gti)\dhat{\vs}_i$, and $K_{Y,i}$ is the Kalman gain. By \eqref{eq:dhat_discussion}, the estimate $\dhat{\vs}_i$ depends on the terms that do not appear in the standard KF. However, the covariance of the estimation error $\hat{\vs}_i - \dhat{\vs}_i$ is not effected by the state-space modifications. We present these results formally.

\begin{lemma}[Directed information]\label{lemma:computing_estimators} 
For a fixed policy of the form as in \eqref{eq:linear_policy}, the directed information can be expressed by the sum
\begin{align}\label{eq:objective}
    I(\vx^n \rightarrow \vy^n) &= \frac{1}{2} \sumT \log \left(\frac{\det \Psi_{Y,i}}{\det \Psi}\right)
\end{align}
where $\Psi_{Y,i}$ is the covariance of the innovation $\psi_i$, and can be computed as 
\begin{align}\label{eq:kalmanY}
    \Psi_{Y,i} = (H+J\Gi \sd)\hat{\Sigma}_i (H+J\Gi \sd)^T + J M_i J^T +  \Psi.
\end{align}
For the state-space in \eqref{eq:observer_model}, the optimal Kalman gain in \eqref{eq:dhat_discussion} is 
\begin{align}\label{eq:Kalman_dec_estimator}
    K_{Y,i}&= ((F+G\Gi\sd)\hat{\Sigma}_i(H+J\Gi\sd)^T + GM_iJ^T +K_p\Psi)\Psi_{Y,i}^{-1}
\end{align}
and its corresponding error covariance $\hat{\Sigma}_i = \mathbf{cov}(\hat{\vs}_i-\dhat{\vs}_i)$ can be computed recursively as
\begin{align}\label{eq:cov_decoder_err}
    \hat{\Sigma}_{i+1} &= (F+G\Gi \sd)\hat{\Sigma}_i(F+G\Gi \sd)^T + GM_iG^T + K_p\Psi K_p^T -K_{Y,i}\Psi_{Y,i}K_{Y,i}^T
\end{align}
with the initial condition $\hat{\Sigma}_1 = 0$.
\end{lemma}


By \eqref{eq:kalmanY} $\Psi_{Y,i}$ does not depend on $\Gamma_{2,i}$, and thus the objective function \eqref{eq:objective} does not depend on $\Gamma_{2,i}$ too. This observation is noteworthy, as also noted in \cite{CC_Bambos_ISIT24}, as it indicates that the role of $\Gamma_{2,i}$ is not to maximize the objective, but to minimize the control cost. In the following lemma, we demonstrate that minimizing the control cost over the set $\{\Gamma_{2,i}\}$ is closely related to the optimal LQR control policy.
\begin{lemma}[Optimal LQR cost]\label{lemma:minimizig_Gamma_2}
The LQR cost in \eqref{eq:intro_LQR_cost}, minimized with respect to $\{\Gti\}_{i=1}^n$, is given by 
\begin{align}\label{eq:cost_non_linear}
    \mathcal{J}_n 
    &= \frac1{n} \Tr(\mathbf{cov}(\dhat{\vs}_1) E_1) + \frac{1}{n}\sumT  \Tr(K_{Y,i} \Psi_{Y,i}K_{Y,i}^TE_{i+1})+\frac1{n}\sum_{i=1}^n \left[\Tr((M_i +\Gamma_{1,i}\hat{\Sigma}^\dagger_i\Gamma_{1,i}^T ) R) \right] \nonumber\\
    &\ \ + \frac1{n}\sum_{i=1}^{n+1}\left[ \Tr (\Sigma Q) + \Tr (\hat{\Sigma}_i Q)\right],
\end{align}
where $E_1,\dots,E_{n+1}$ are computed from the LQR control Riccati recursion in \eqref{eq:control_recursion}. Moreover, the minimizer at time $i$ is $\Gti = - K_{\mathtt{LQR},i}$, defined as $\KLQRi = (R+G^TE_{i+1}G)^{-1}G^TE_{i+1}F$.
\end{lemma}

Combining Lemma \ref{lemma:minimizig_Gamma_2} with Lemma \ref{lemma:Gaussian_linear_policy} shows that an optimal policy for the control-coding problem is 
\begin{align}
    \vx_i&= \Gi \sd \Err{i} - K_{\mathtt{LQR},i} \dhat{\vs}_i + \mathbf{m}_i, \ \ i=1,\dots,n,
\end{align}
where $\Gamma_{1,i}$ and $M_i$ are the decision variables to be optimized. As $\Gamma_{2,i}$ no longer appears, we use $\Gamma_i$ as a shorthand for $\Gamma_{1,i}$. This is a significant simplification, but it is still given as a multi-letter expression. Our next steps are to convert the limiting multi-letter to a single-letter (finite-dimensional) optimization problem through convexification of an upper bound to the $n$-letter problem. In particular, both the objective and the control cost constraint should be written as an affine mappings of the decision variables, but at the cost of relaxing a Riccati recursion (given by a set of equalities) to Riccati inequalities. The first step is to bound $C_n(p)$ by a sequential convex optimization problem (SCOP).

\begin{lemma}\label{lemma:T-letter_SDP}
The $n$-letter control-coding capacity in \eqref{eq:n_letter_capacity} is upper bounded by the SCOP
\begin{align}\label{eq:scop_capacity}
    C_n(p) \leq &\max_{\{\Pi_i,\Gamma_i, \hat{\Sigma}_{i+1}\}_{i=1}^n}  \frac{1}{2n}\sumT\log \det (\Psi_{Y,i}) - \log \det (\Psi)\\
        s.t. \ \ & \frac1{n}\sum_{i=1}^n \left[ \Tr( \hat{\Sigma}_i\KLQRi^T \Psi_{\mathtt{LQR},i}\KLQRi) + \Tr(\Pi_i  \Psi_{\mathtt{LQR},i}) +\Tr(K_p\Psi K_p^TE_{i+1}) \right]\nonumber\\
    &\ \ + 2\frac1{n}\sumT \Tr(\Gamma_i \KLQRi^T \Psi_{\mathtt{LQR},i}) + \Tr (\Sigma Q) + \mathcal{E}_n  \leq p \nonumber
    \\&\Psi_{Y,i} = J\Pi_i J^T +H\hat{\Sigma}_i H^T+H\Gamma_i^TJ^T+ J\Gamma_i H^T+ \Psi \nonumber \\ 
    &  K_{Y,i}=  (F\Gamma_i J^T + F \hat{\Sigma}_i H^T + G\Pi_i J^T + G \Gamma_i H^T +K_{p}\Psi) \Psi_{Y,i}^{-1} \nonumber \\&
    \begin{pmatrix}
        \Pi_i & \Gamma_i\\  \Gamma_i^T & \hat{\Sigma}_i
    \end{pmatrix} \succeq 0, \quad \hat{\Sigma}_{T+1} \succeq 0 \nonumber \\&
    \begin{pmatrix}
     F\hat{\Sigma}_i F^T+F\Gamma_i G^T+ G\Gamma_i F^T + G\Pi_i G^T+  K_{p}\Psi K_{p}^T - \hat{\Sigma}_{i+1}& K_{Y,i}\Psi_{Y,i}\\ 
      \Psi_{Y,i}K_{Y,i}^T & \Psi_{Y,i}
    \end{pmatrix} \succeq 0, \nonumber 
\end{align}
where $\mathcal E_n = \frac1{n}\left(\Tr (( \Sigma + \hat{\Sigma}_{n+1}) Q) + \Tr(\mathbf{cov}(\dhat{\vs}_1) E_1) + \Tr(\hat{\Sigma}_1E_1 - \hat{\Sigma}_{n+1}E_{n+1})\right)$, the constraints hold for $i=1, \dots , T$, $\hat{\Sigma}_1 =0$, and $ \KLQRi, \Psi_{\mathtt{LQR},i},E_i,$ are defined in \eqref{eq:control_recursion}. 
\end{lemma}

The significance of Lemma \ref{lemma:T-letter_SDP} is that the objective is convex and the constraints are linear functions of the decision variables. We name this convex optimization problem as a SCOP since the second LMI constraint involves two subsequent decision variables $\hat \Sigma_i$ and $\hat \Sigma_{i+1}$. 

\subsection{Proofs of the Theorems \ref{theorem:upper_bound}-\ref{theorem:scalar}}\label{subsec:theorems_proof}
In the next proof we utilize the convexity, along with the limit computation, to establish the single-letter upper bound in Theorem \ref{theorem:upper_bound}. 
\begin{proof}[Proof of Theorem \ref{theorem:upper_bound}]
Throughout the derivations so far, we studied the $n$-letter capacity $C_{n}(p)$ in \eqref{eq:scop_capacity}. A standard converse argument can relate this quantity to the feedback capacity by showing that for any~$n$,
\begin{align}
    C_{\mathtt{LQG}}(p)\leq \frac1{n}C_{n}(p) + \delta_n,
\end{align}
where $\delta_n \rightarrow 0$ is resulted from a Fano's inequality. As Lemma \ref{lemma:T-letter_SDP} provides a bound on $\frac1{n}C_{n}(p)$ in terms of a SCOP, the remaining step is to show that the SCOP can be further upper bounded by its single-letter counterpart in Theorem \ref{theorem:upper_bound}.

Define the averaged decision variables as
\begin{align}\label{eq:convex_combinations}
    \overline{\Pi}_n= \frac{1}{n}\sum_{i=1}^n \Pi_i, \overline{\Gamma}_n= \frac{1}{n} \sum_{i=1}^n\Gamma_i, \overline{\hat{\Sigma}}_n= \frac{1}{n} \sum_{i=1}^n\hat{\Sigma}_i   
\end{align}

We start by showing the objective function is upper bounded by the single-letter evaluated at the averaged variables. Consider the obejctive function in Lemma \ref{lemma:T-letter_SDP},
\begin{align}
    \frac{1}{n}\sum_{i=1}^n \log \det (\Psi_{Y,i}) &= \frac{1}{n}\sum_{i=1}^n \log \det (J\Pi_i J^T +H\hat{\Sigma}_i H^T+H \Gamma_i^TJ^T + J\Gamma_i H^T+ \Psi) \nonumber\\
    &\leq \log \det\left(J\overline{\Pi}_n J^T +H\overline{\hat{\Sigma}}_n H^T+H \overline{\Gamma}_n^TJ^T + J\overline{\Gamma}_n H^T +\Psi\right) \nonumber\\
    &\triangleq \log \det\left(\overline{\Psi}_{Y,n}\right),
\end{align}
where the inequality follows from Jensen's inequality by utilizing the concavity of $\log \det ( \cdot )$ and the fact that the argument is an affine function of the decision variables.

We proceed to show that the decision variables satisfy the single-letter constraints that appear in the optimization of \eqref{eq:main_UB}:
\begin{align*}
     &\begin{pmatrix}
        \overline{\Pi}_n & \overline{\Gamma}_n\\  \overline{\Gamma}_n^T & \overline{\hat{\Sigma}}_n
    \end{pmatrix} \succeq 0 \\
    &\Tr( \overline{\hat{\Sigma}}_n K_{\mathtt{LQR}}^T \Psi_{\mathtt{LQR}}K_{\mathtt{LQR}}) + \Tr(\overline{\Pi}_n  \Psi_{\mathtt{LQR}}) +\Tr(K_p\Psi K_p^TE) + 2 \Tr( \overline{\Gamma}_n K_{\mathtt{LQR}}^T \Psi_{\mathtt{LQR}}) + \Tr (\Sigma Q)  \leq p \\
    &\begin{pmatrix}
     F\overline{\hat{\Sigma}}_n F^T+F\overline{\Gamma}_n G^T+ G\overline{\Gamma}_n F^T + G\overline{\Pi}_n G^T+  K_{p}\Psi K_{p}^T - \overline{\hat{\Sigma}}_{n}& \overline{K_{Y,n}\Psi_{Y,n}}\\ 
    \overline{K_{Y,n}\Psi_{Y,n}}^T & \overline{\Psi}_{Y,n}
    \end{pmatrix} \succeq 0. 
\end{align*}
where $\overline{K_{Y,n}\Psi_{Y,n}}= F\overline{\Gamma}_n J^T + F \overline{\hat{\Sigma}}_n H^T + G\overline{\Pi}_n J^T + G \overline{\Gamma}_n H^T +K_{p}\Psi$ and $\overline{\Psi}_{Y,n}= H\overline{\Gamma}_n J^T + H \overline{\hat{\Sigma}}_n H^T + J\overline{\Pi}_n J^T + J \overline{\Gamma}_n H^T +\Psi$. 

We will show that these constraints are satisfied by studying the per-time constraints of the $n$-letter optimization problem. For the first LMI constraint, the single-letter constraint can be directly written as the average of the per-time constraints
\begin{align}
    \begin{pmatrix}
        \overline{\Pi}_n & \overline{\Gamma}_n\\  \overline{\Gamma}_n^T & \overline{\hat{\Sigma}}_n
    \end{pmatrix} = \frac1{n} \sum_{i=1}^n
    \begin{pmatrix}
        \Pi_i & \Gamma_i\\ \Gamma_i^T & \hat{\Sigma}_i
    \end{pmatrix} \succeq 0. 
\end{align}
Similarly the Riccati LMI constraint is satisfied at the convex combination
\begin{align}
    &0\preceq\frac1{n}\sumT 
     \begin{pmatrix}
     F\hat{\Sigma}_i F^T+F\Gamma_i G^T+ G\Gamma_i F^T + G\Pi_i G^T+  K_{p}\Psi K_{p}^T - \hat{\Sigma}_{i+1}& K_{Y,i}\Psi_{Y,i}\\ 
     \Psi_{Y,i} K_{Y,i}^T & \Psi_{Y,i}
    \end{pmatrix} \nonumber \\
    \label{eq:riccati_convex}&=\begin{pmatrix}
     F\overline{\hat{\Sigma}}_n F^T+F\overline{\Gamma}_n G^T+ G\overline{\Gamma}_n F^T + G\overline{\Pi}_n G^T+  K_{p}\Psi K_{p}^T - \overline{\hat{\Sigma}}_{n}& \overline{K_{Y,n}\Psi_{Y,n}}\\ 
    \overline{K_{Y,n}\Psi_{Y,n}}^T & \overline{\Psi}_{Y,n}
    \end{pmatrix} \\
    &- \begin{pmatrix}
        \frac1{n}(\hat{\Sigma}_{n+1}-\hat{\Sigma}_1)& 0 \\0&0
    \end{pmatrix}  \nonumber
\end{align}
since $\hat{\Sigma}_1=0, \hat{\Sigma}_{n+1} \succeq 0 $ we conclude that also the Riccati LMI on the convex combination in \eqref{eq:riccati_convex} is also P.S.D. 

The last constraint is the LQR control cost, we first prove two simple claims.
\underline{Claim 1:} let $A,B$ matrices and $||A||_F < \infty, ||B||_F\rightarrow 0$ then $\Tr(AB) \rightarrow 0$. \underline{Proof:} $|\Tr(AB)| =|<A^T,B>|\leq ||A||_F  ||B||_F \rightarrow 0$.\\
\underline{Claim 2:} let $\{A_i\}\rightarrow A, \{B_i\}\rightarrow B$, then $\frac{1}{n} \sumT\Tr(A_i B_i) = \Tr(AB) + \epsilon_n$ where $\epsilon_n \to 0$.\\
\underline{Proof:} $\frac{1}{n} \sumT \Tr(A_i B_i) = \frac{1}{n} \Tr((A_i + A-A) (B_i+B-B)) = \Tr(AB) + \frac{1}{n} \sumT \Tr(A(B_i-B))+\Tr((A_i-A)B) + \Tr((A_i -A)(B_i-B)) = \Tr(AB) + \epsilon_n$ where $\epsilon_n \to 0$. 
Let $\overline{K}_{\mathtt{LQR},n} \stackrel{\Delta}= \frac1{n} \sum_{i=1}^n K_{\mathtt{LQR},i} ,\overline{\Psi}_{\mathtt{LQR},n} \stackrel{\Delta}= \frac1{n} \sum_{i=1}^n \Psi_{\mathtt{LQR},i}$ denote the averaged constants of the LQR variables in Section \ref{subsec:lqr}, and under the assumptions, $\overline{\Psi}_{\mathtt{LQR},n} \rightarrow \Psi_{\mathtt{LQR}}$ and $\overline{K}_{\mathtt{LQR},n} \rightarrow K_{\mathtt{LQR}}$, and therefore $\overline{K_{\mathtt{LQR},n}^T \Psi_{\mathtt{LQR},n}} =\frac1{n} \sum_{i=1}^n  K_{\mathtt{LQR},i}^T\Psi_{\mathtt{LQR},i}$ and $K_{\mathtt{LQR},n}^T\Psi_{\mathtt{LQR},n}\to  K_{\mathtt{LQR}}^T\Psi_{\mathtt{LQR}}$.
We utilize the control constraint in Lemma \ref{lemma:T-letter_SDP} to show that the convex combination of the decision variables adhere to the control budget $p$ as $n$ grows. 
\begin{align*}
p&\ge \frac1{n}\sum_{i=1}^n \left[ \Tr( \hat{\Sigma}_i\KLQRi^T \Psi_{\mathtt{LQR},i}\KLQRi) + \Tr(\Pi_i  \Psi_{\mathtt{LQR},i}) +\Tr(K_p\Psi K_p^TE_{i+1}) \right]\\
    &\ \ + 2\frac1{n}\sumT \Tr(\Gamma_i \KLQRi^T \Psi_{\mathtt{LQR},i}) + \Tr (\Sigma Q) + \mathcal{E}_n \\
   &=\Tr( \overline{\hat{\Sigma}}_n K_{\mathtt{LQR}}^T \Psi_{\mathtt{LQR}}K_{\mathtt{LQR}}) + \Tr(\overline{\Pi}_n  \Psi_{\mathtt{LQR}}) +\Tr(K_p\Psi K_p^TE) + 2 \Tr( \overline{\Gamma}_n K_{\mathtt{LQR}}^T \Psi_{\mathtt{LQR}}) + \Tr (\Sigma Q) + \epsilon_n' 
\end{align*}
where the inequality follows from Lemma \ref{lemma:T-letter_SDP}, and the equality follows from Claim $2$ with $\epsilon_n' := \mathcal{E}_n+\epsilon_n  \rightarrow 0$. This shows that the LQR cost constraint is satisfied asymptotically. The compactness of the optimization domain implies that there exists a limiting point that satisfies the constraints.
\end{proof}

\begin{proof}[Proof of Theorem \ref{theorem:achievabilty}]
To show the tightness of the upper bound, we introduce a general lower bound that is  implied from Lemmas \ref{lemma:Gaussian_linear_policy} - \ref{lemma:minimizig_Gamma_2}. We will then utilize this lower bound with the optimal parameters from the upper bound. 

For a time-invariant policy, given by the pair $(\overline{\Gamma},M)$, we have the state-space model from Lemma \ref{lemma:Gaussian_linear_policy}
\begin{align*}
    \hat{\vs}_{i+1} &= (F+G\overline{\Gamma}) \Err{i} + (F-K_{\mathtt{LQR}}G ) \dhat{\vs}_i + G\mathbf{m}_i + K_p \mathbf{e}_i \nonumber\\
    \vy_i&=(H+J\overline{\Gamma}) \Err{i} + (H-K_{\mathtt{LQR}}J)\dhat{\vs}_i + J \mathbf{m}_i  + \mathbf{e}_i.
\end{align*}
This is a standard state-space model that obeys the Riccati recursion 
\begin{align}\label{eq:LB_rec_hat_sigma}
    \hat{\Sigma}_{i+1} = (F+G\overline{\Gamma} )\hat{\Sigma}_i(F+G\overline{\Gamma})^T + GMG^T + K_p\Psi K_p^T -\overline{K}_{Y,i}\overline{\Psi}_{Y,i}\overline{K}_{Y,i}^T
\end{align}
where $\overline{K}_{Y,i} = ((F+G\overline{\Gamma})\hat{\Sigma}_i(H+J\overline{\Gamma})^T + GMJ^T +K_p\Psi)\overline{\Psi}_{Y,i}^{-1}$ and $\overline{\Psi}_{Y, i} = (H+J\overline{\Gamma})\hat{\Sigma}_i (H+J\overline{\Gamma})^T + J M J^T +  \Psi $.

If $\hat{\Sigma}_i$ in \eqref{eq:LB_rec_hat_sigma} converges to a stationary point, we can compute it by solving the Riccati equation 
\begin{align}\label{eq:LB_ricatti_hat_sigma}
    \hat{\Sigma}= (F+G \overline{\Gamma})\hat{\Sigma}(F+G \overline{\Gamma})^T + GMG^T + K_p\Psi K_p^T -\overline{K}_{Y}\overline{\Psi}_{Y}\overline{K}_{Y}^T
\end{align}
with $\overline{K}_Y, \overline{\Psi}_Y$ are as before, but evaluated at the stationary point $\hat\Sigma$. This convergence is guaranteed under standard regularity conditions, see Appendix \ref{app:sub_hat_sigma}.

Since the convergence has a geometric rate, by Lemma \ref{lemma:minimizig_Gamma_2} and Lemma \ref{lemma:computing_estimators}, a policy whose Riccati recursion \eqref{eq:LB_rec_hat_sigma} converges to $\hat\Sigma$ achieves a lower bound on the LQG system capacity
\begin{align*}
    C(p^*)\ge \frac12\log\det(\overline{\Psi}_{Y}) - \frac12\log\det(\Psi) 
\end{align*}
where $\overline{\Psi}_{Y} = (H+J\overline{\Gamma})\hat{\Sigma} (H+J\overline{\Gamma})^T + J M J^T +  \Psi$ and the control budget is equal to 
\begin{align*}
    p^*&=\Tr( \hat{\Sigma}K_{\mathtt{LQR}}^T \Psi_{\mathtt{LQR}}K_{\mathtt{LQR}}) + \Tr(\Pi  \Psi_{\mathtt{LQR}}) +\Tr(K_p\Psi K_p^TE) + 2 \Tr(\overline{\Gamma} K_{\mathtt{LQR}}^T \Psi_{\mathtt{LQR}}) + \Tr (\Sigma Q).
\end{align*}

We proceed to prove Theorem \ref{theorem:achievabilty} based on the lower bound with the particular parameters
\begin{align}\label{eq:proof_2_param}
    \overline{\Gamma}&=\Gamma^* \hat{\Sigma}^{*\dagger} \nonumber\\
    M& = \Pi^* - \Gamma^*\hat{\Sigma}^{*\dagger}\Gamma^{*T}
\end{align}
where $\Gamma^*, \hat{\Sigma}^*, \Pi^*$ are the optimal decision variables from the optimization \eqref{eq:main_UB}.
The regularity conditions in Theorem \ref{theorem:achievabilty} guarantee that the Riccati equation \eqref{eq:LB_ricatti_hat_sigma} has a unique solution, denoted as $\hat{\Sigma}_{LB}$.

We proceed to show that $\hat \Sigma^*$ is also a solution to the Riccati equation, and therefore, $\hat \Sigma_{LB} = \hat \Sigma^*$. To simplify the derivation, we denote the Riccati equation as $f(\hat \Sigma)$ with $f(\hat \Sigma_{LB})=0$.
\begin{align}
    f(\hat \Sigma^*) &= (F+G \overline{\Gamma})\hat{\Sigma}^*(F+G \overline{\Gamma})^T + GMG^T + K_p\Psi K_p^T -\overline{K}_{Y}\overline{\Psi}_{Y}\overline{K}_{Y}^{T} - \hat{\Sigma}^* \nonumber\\
    &\stackrel{(a)}= (F+G \Gamma^* \hat{\Sigma}^{*\dagger})\hat{\Sigma}^*(F+G \Gamma^* \hat{\Sigma}^{*\dagger})^T + G(\Pi^* - \Gamma^*\hat{\Sigma}^{*\dagger}\Gamma^{*T})G^T + K_p\Psi K_p^T -\overline{K}_{Y}\overline{\Psi}_{Y}\overline{K}_{Y}^{T} - \hat{\Sigma}^* \nonumber\\
    &\stackrel{(b)}=F\hat{\Sigma}^* F^T+F\Gamma^{*T}G^T+ G\Gamma^* F^T + G\Pi^* G^T+  K_{p}\Psi K_{p}^T - K_{Y}^*\Psi_{Y}^*K_{Y}^{*T}- \hat{\Sigma}^* \nonumber\\
    &\stackrel{(c)}=F\hat{\Sigma}^* F^T+F\Gamma^{*T}G^T+ G\Gamma^* F^T + G\Pi^* G^T+  K_{p}\Psi K_{p}^T - K_{Y}^*\Psi_{Y}^*K_{Y}^{*T}- \hat{\Sigma}^* \nonumber\\
    &\stackrel{(d)}= 0,
\end{align}
where $(a)$ follows from the policy parameters in \eqref{eq:proof_2_param}, $(b)$ follows from the first LMI constraint in \eqref{eq:main_UB} implying that $\Gamma(I- \hat \Sigma^\dagger \hat \Sigma)=0$, $(c)$ follows from 
\begin{align*}
    \overline{K}_{Y} &= ((F+G\Gamma^* \hat{\Sigma}^{*\dagger}) \hat{\Sigma}^*(H+J\Gamma^* \hat{\Sigma}^{*\dagger})^T + GMJ^T +K_p\Psi)\Psi_{Y}^{*-1} \\&=(F\Gamma^* J^T + F \hat{\Sigma}^* H^T + G\Pi^* J^T + G \Gamma^* H^T +K_{p}\Psi) \Psi_{Y}^{*-1}\\ 
    &= K_Y^*,
\end{align*}
and $(d)$ follows from the assumed sufficient condition in Theorem \ref{theorem:achievabilty} that the Riccati equation  \eqref{eq:main_ricc_eq_sufficient} is satisfied. Since the Riccati equation $f(\cdot)$ has a unique PSD solution, we conclude that $\hat{\Sigma}^* = \hat\Sigma_{LB}$.

\end{proof}
\begin{proof}[Proof of Theorem \ref{theorem:scalar}]
The main idea of the proof is to show via the KKT conditions that the optimal solution of the upper bound satisfies the sufficient condition in Theorem \ref{theorem:achievabilty}.

In the scalar case, the maximization of $\log \det(\Psi_Y) - \log \det(\Psi)$ can be reduced to the maximization of $\Psi_Y$, i.e., the optimization \eqref{eq:main_UB} is equivalent to the optimization of $\Psi_Y = \exp(C_{\mathtt{LQG}}(p))$ given by
    \begin{align}\label{eq:scalar_optimization}
    &\max_{\Pi,\Gamma, \hat{\Sigma}} H^2 \hat{\Sigma} + J^2\Pi + 2 J H \Gamma +\Psi\\
        s.t. \quad  
    &g_1(\hat{\Sigma}, \Pi, \Gamma) :=
    \begin{pmatrix}
        \Pi & \Gamma \\ \Gamma &\hat{\Sigma}
    \end{pmatrix}\succeq 0\nonumber \\
     &g_2(\hat{\Sigma}, \Pi, \Gamma) := p-\Sigma Q-  K_p^2\Psi E - K_{\mathtt{LQR}}^2\Psi_{\mathtt{LQR}}\hat{\Sigma}-\Psi_{\mathtt{LQR}}\Pi  - 2K_{\mathtt{LQR}}\Psi_{\mathtt{LQR}}\Gamma \geq 0 \nonumber\\
     &g_3(\hat{\Sigma}, \Pi, \Gamma) :=((F^2 -1)\hat{\Sigma}+ 2FG\Gamma+ G^2\Pi +K_p^2\Psi) -\frac{(F \hat{\Sigma} H + GJ\Pi +(FJ+GH)\Gamma +K_p\Psi)^2}{ H^2 \hat{\Sigma} + J^2\Pi + 2 J H \Gamma +\Psi} \geq 0, \nonumber
    \end{align}
    where $K_{\mathtt{LQR}}$, $\Psi_{\mathtt{LQR}}$ and $E$ are the (scalar) constants defined in \eqref{eq:riccati_sigma}, $g_3$ is the Schur complement of the second two-dimensional LMI constraint in \eqref{eq:main_UB} in the scalar case. 

    The Lagrangian of \eqref{eq:scalar_optimization} is given by 
    \begin{align*}
     L(\Pi, &\Gamma, \hat{\Sigma}, \Lambda_{1},\lambda_{2},\lambda_{3}) = H^2 \hat{\Sigma} + J^2\Pi + 2 J H \Gamma + \Psi+ \Tr(g_1 \Lambda_1) +g_2 \lambda_2 +g_3 \lambda_3 
    \end{align*}
    where $ \Lambda_1 = 
    \begin{pmatrix}
        \Lambda_{1,1} & \Lambda_{1,2} \\ \Lambda_{1,2} & \Lambda_{2,2}
    \end{pmatrix}\succeq 0, \lambda_2 \geq 0 $ and $\lambda_3 \geq 0$ are the dual variables. 
    
    By the KKT optimality conditions, the stationarity constraints are given by
\begin{align}\label{eq:proof_stationary1}
    &\frac{\partial L}{\partial \Pi} = J^2+ \Lambda_{1,1} - \Psi_{\mathtt{LQR}}\lambda_2 + \frac{\partial g_3}{\partial \Pi}\lambda_3 = 0 \nonumber\\
    &\frac{\partial L}{\partial \Gamma} = 2HJ +2\Lambda_{1,2} -2K_{\mathtt{LQR}}\Psi_{\mathtt{LQR}}\lambda_2  +\frac{\partial g_3}{\partial \Gamma}\lambda_3  = 0 \nonumber\\
    &\frac{\partial L}{\partial \hat{\Sigma}} = H^2 + \Lambda_{2,2} -K_{\mathtt{LQR}}^2\Psi_{\mathtt{LQR}}\lambda_2 + \frac{\partial g_3}{\partial \hat{\Sigma}}\lambda_3 = 0.
\end{align} 
    and the complementary slackness conditions are
    \begin{align*}
    &g_1(\hat{\Sigma}^*, \Pi^*, \Gamma^*)\Lambda_1=0 \\
    &g_2(\hat{\Sigma}^*, \Pi^*, \Gamma^*)\lambda_2 =0 \\
    &g_3(\hat{\Sigma}^*, \Pi^*, \Gamma^*)\lambda_3 =0,
\end{align*}

In the first part of the proof we show that the sufficient condition from Theorem \ref{theorem:achievabilty} is satisfied. The Riccati equation in this case is given by $g_3(\hat{\Sigma}^*, \Pi^*, \Gamma^*)=0$. Assume by contradiction that $g_3(\hat{\Sigma}^*, \Pi^*, \Gamma^*)>0$, and therefore by the complementary slackness condition $\lambda_3=0$. 
We rule out the case $\hat{\Sigma} =0$: if $\hat{\Sigma} =0$, the feasibility constraint of $g_1$ implies $\Gamma=0$. If $\Pi = 0$, the objective is zero, and we obtain a contradiction to the assumption that the capacity is positive. Thus we can assume $\Pi > 0$. By the complementary slackness condition on $g_1$, 
$$\begin{pmatrix}
        \Pi & \Gamma \\ \Gamma &\hat{\Sigma}
    \end{pmatrix}\begin{pmatrix}
        \Lambda_{1,1} & \Lambda_{1,2} \\ \Lambda_{1,2} & \Lambda_{2,2}
    \end{pmatrix} = \begin{pmatrix}
        \Pi\Lambda_{1,1}+\Gamma \Lambda_{1,2}  & \Pi\Lambda_{1,2}++\Gamma \Lambda_{2,2} \\ \Gamma\Lambda_{1,1}+ \hat{\Sigma}\Lambda_{1,2} & \Gamma\Lambda_{1,2}+ \hat{\Sigma}\Lambda_{2,2}
    \end{pmatrix} = 0 $$
and since $\Pi>0$ we get $\Lambda_{1,1}=\Lambda_{1,2} = 0$. Substituting these in the stationarity conditions in \eqref{eq:proof_stationary1} provides with
\begin{align*}
    0&= J^2  - \Psi_{\mathtt{LQR}}\lambda_2 \\
    0&= 2HJ  -2K_{\mathtt{LQR}}\Psi_{\mathtt{LQR}}\lambda_2   \\
    0&= H^2 + \lambda_{2,2} -K_{\mathtt{LQR}}^2\Psi_{\mathtt{LQR}}\lambda_2.
\end{align*}
By substituting $J^2  = \Psi_{\mathtt{LQR}}\lambda_2$ into the second equation, we get $J(H-K_{\mathtt{LQR}}J)=0$, since we assumed that $H \neq K_{\mathtt{LQR}}J$ it implies $J=0$, but then $H=0$ and the capacity is zero. All in all, we showed $\hat{\Sigma} \neq 0$.

As $\hat{\Sigma} \neq 0$, it is convenient to replace the first LMI constraint with its Schur complement. In particular, we replace the LMI constraint $g_1 \succeq 0$ in \eqref{eq:scalar_optimization} with the scalar constraints
\begin{align*}
     g_4(\hat{\Sigma}, \Pi, \Gamma) &=\quad \hat{\Sigma} >0\\
     g_5(\hat{\Sigma}, \Pi, \Gamma) &=\Pi - \frac{\Gamma^2}{\hat{\Sigma}} \geq 0,
\end{align*}
The modified Lagrangian is then 
\begin{align*}
     L(\Pi, &\Gamma, \hat{\Sigma}, \lambda_{2},\lambda_{3}, \lambda_4,\lambda_{5}) = H^2 \hat{\Sigma} + J^2\Pi + 2 J H \Gamma +\Psi +  g_2 \lambda_2 +g_3 \lambda_3+\hat{\Sigma}\lambda_4 +g_5 \lambda_5,
\end{align*}
where $\lambda_4 \geq 0, \lambda_5 \geq 0$ are the additional dual variables. 

The stationarity conditions under the assumption $\lambda_3=0$, are
\begin{align}\label{eq:stationarity}
    &\frac{\partial L}{\partial \Pi} = J^2+ \lambda_5 - \Psi_{\mathtt{LQR}}\lambda_2 + (G-K_YJ)^2\lambda_3 = 0 \nonumber\\
    &\frac{\partial L}{\partial \Gamma} = 2HJ -2\frac{\Gamma}{\hat{\Sigma}}\lambda_5 -2K_{\mathtt{LQR}}\Psi_{\mathtt{LQR}}\lambda_2 + 2(G-K_YJ)(F-K_YH)\lambda_3 = 0 \nonumber\\
    &\frac{\partial L}{\partial \hat{\Sigma}} = H^2 + \lambda_4 + \frac{\Gamma^2}{\hat{\Sigma}^2}\lambda_5 -K_{\mathtt{LQR}}^2\Psi_{\mathtt{LQR}}\lambda_2  + ((F-K_YH)^2-1)\lambda_3  = 0.
\end{align}
The complementary slackness conditions are 
\begin{align*}
    &g_i(\hat{\Sigma}^*, \Pi^*, \Gamma^*)\lambda_i =0, \ \ \ i=2,\dots,5. 
\end{align*}
We showed $\hat{\Sigma} \neq0$ and therefore by the complementary slackness we have $\lambda_4=0$.
By the first stationarity condition in \eqref{eq:stationarity}, we have $J^2 + \lambda_5=\Psi_{\mathtt{LQR}}\lambda_2$. We show now $J\neq0$. If $J=0$, then $\lambda_5 = \Psi_{\mathtt{LQR}}\lambda_2$ and by the second stationarity condition $\frac{\Gamma}{\hat{\Sigma}}\lambda_5  =- K_{\mathtt{LQR}}\lambda_5$. By the third stationarity condition we obtain $H^2 = 0$ so in this case the capacity is zero.
Assume that $\lambda_5 =0$, then $J^2=\Psi_{\mathtt{LQR}}\lambda_2$, then by the second stationarity condition $J(H-K_{\mathtt{LQR}}J) = 0$, in contradiction to our assumptions, $J \neq 0, H\neq K_{\mathtt{LQR}}J $, so we can conclude that $\lambda_5 > 0$. 
Therefore by the complementary slackness conditions
\begin{align}\label{eq:pi}
    &\Pi = \frac{\Gamma^2}{\hat{\Sigma}}.
\end{align}
Then we can rewrite the second stationarity condition in \eqref{eq:stationarity} as 
\begin{align}\label{eq:sec_zeta}
    HJ - \zeta \lambda_5 -K_{\mathtt{LQR}}(J^2 +\lambda_5) =0
\end{align}
where $\zeta = \frac{\Gamma}{\hat{\Sigma}}$.
Solving for $\lambda_5$ we have 
\begin{align*}
    \lambda_5 = \frac{J(H-K_{\mathtt{LQR}}J)}{K_{\mathtt{LQR}}+\zeta}
\end{align*}
The third stationarity condition in \eqref{eq:stationarity} can be written as
\begin{align*}
    &H^2 + \zeta^2 \lambda_5 -K_{\mathtt{LQR}}^2(J^2 +\lambda_5) =0\\
    &\to \lambda_5 = \frac{H^2-(K_{\mathtt{LQR}}J)^2}{K_{\mathtt{LQR}}^2-\zeta^2}
\end{align*}
Notice that if $K_{\mathtt{LQR}}^2 = \zeta^2$, then we have $H^2= (K_{\mathtt{LQR}}J)^2$, since $H \neq K_{\mathtt{LQR}}J$ by our assumption, we just assume $H = -K_{\mathtt{LQR}}J$, but then in \eqref{eq:sec_zeta} we have $-K_{\mathtt{LQR}}J^2 = (K_{\mathtt{LQR}}+\zeta) \lambda_5$, if $K_{\mathtt{LQR}} = - \zeta$ then $J=0$ and if $K_{\mathtt{LQR}} = \zeta$ then $J^2 = -2\lambda_5$ as contradiction to positivity of $\lambda_5$.

by comparing $\lambda_5$ we have
\begin{align*}
    &H^2-(K_{\mathtt{LQR}}J)^2) = J(H-K_{\mathtt{LQR}}J)(K_{\mathtt{LQR}}-\zeta)\\
    & \to \frac{H+K_{\mathtt{LQR}}J}{J} = K_{\mathtt{LQR}}-\zeta\\
    & \to \zeta = \frac{K_{\mathtt{LQR}}J-(H+K_{\mathtt{LQR}}J)}{J}=-\frac{H}{J}
\end{align*}
Recall that $\zeta = \frac{\Gamma}{\hat{\Sigma}} $ then $\Gamma = \frac{-H \hat{\Sigma}}{J}$ and $\Pi = \frac{H^2}{J^2}\hat{\Sigma}$,
then evaluating $\Psi_Y = H^2 \hat{\Sigma} + J^2\Pi + 2 J H \Gamma +\Psi = \Psi $ and the capacity is zero. This completes the proof that the sufficient condition in Theorem 2 is satisfied. We proceed to show the regularity conditions for the scalar case.

\underline{Proof of detectability:}
Recall the lower bound definitions $\overline{\Gamma} = \Gamma^* \hat{\Sigma}^{*\dagger}$ and $M = \Pi^* - \Gamma^*\hat{\Sigma}^{*\dagger}\Gamma^{*T}$ where $\Gamma^*,\hat{\Sigma}^{*\dagger}$ are the optimal variables from the upper bound. We show using contradiction the detectability of $(F + G \overline{\Gamma}, H+J\overline{\Gamma})$, i.e., if $|F + G \overline{\Gamma}| \ge1$ then $H+J\overline{\Gamma} \neq 0$.
We showed above that $g_3 = 0$ and by substituting  $\Pi^*$ with  $\Gamma^*\hat{\Sigma}^{*\dagger}\Gamma^{*T}+M$ and define $\overline{\Gamma} = \Gamma^* \hat{\Sigma}^{*\dagger}$ we have
\begin{align}\label{eq:proof_detecta_Ricc}
    (F+G\overline{\Gamma} )^2\hat{\Sigma}-\hat{\Sigma} + G^2M + K_p^2\Psi  -\overline{K}_{Y}^2\overline{\Psi}_{Y} = 0,
\end{align}
where $\overline{K}_{Y} = ((F+G\overline{\Gamma})(H+J\overline{\Gamma})\hat{\Sigma} + GMJ +K_p\Psi)\overline{\Psi}_{Y}^{-1}$, $\overline{\Psi}_{Y} = (H+J\overline{\Gamma})^2\hat{\Sigma}+ J^2 M  +  \Psi$.
Assume by contradiction $F + G \overline{\Gamma} \ge 1, H+J\overline{\Gamma} = 0$ then the capacity as in \eqref{eq:objective} is zero unless $M>0$. But the equation in \eqref{eq:proof_detecta_Ricc} using $H+J\overline{\Gamma} = 0$ can be written as 
\begin{align*}
       \frac{(F+G\overline{\Gamma} )^2-1}{J^2M+\Psi} \hat{\Sigma} +\frac{(G-K_pJ)^2 M \Psi}{J^2M+\Psi} = 0. 
\end{align*}
Both terms are non-negative, so each of them must be zero. However, the second term is positive since $M>0$ and our assumption that $G \neq K_pJ$ leading to a contradiction. 

\underline{Convergence of the Riccati recursion:} Having established the detectability of the Riccati recursion parameters, there exists a maximal solution to the Riccati equation. We now need to ensure that the Riccati recursion of the lower bound as in \eqref{eq:LB_rec_hat_sigma}, for the scalar case, converges to $\hat \Sigma^*$ from the upper bound when the initial condition is $\hat \Sigma_1=0$. We divide the proof to two cases: if $M>0$, we show that the pair $(F^s, G^s W^{s})$ is stabilizable and by Theorem \ref{theorem:achievabilty} the upper is equal to the LQG system capacity. Indeed, in this case there is a unique solution to the Riccati equation and we utilize this fact to show that the limit from the lower bound coincides with solution $\hat \Sigma$ from the upper bound. For the other case, $M=0$, stabilizability does not hold, and therefore, we study the Riccati recursion directly. We show that choosing $M_1>0$ at the first time followed by the time-invariant policy from the upper bound (with $M=0$) guarantees that it converges to the maximal solution of the Riccati equation. 

We first show that, for $M>0$, $(F^s,G^sW^s)$ is stabilizable where
\begin{align*}
    &F^s = F+G\overline{\Gamma}-(GMJ+K_p\Psi)(JMJ+\Psi)^{-1}(H+J\overline{\Gamma})\\
    &W^s = \begin{pmatrix}
  M & 0\\ 
  0 & \Psi
\end{pmatrix} - 
\begin{pmatrix}
  MJ \\ 
  \Psi
\end{pmatrix} (JMJ+\Psi)^{-1}
\begin{pmatrix}
  JM &\Psi
\end{pmatrix} \nonumber\\
    &G^s=\begin{pmatrix} G &K_p \end{pmatrix}\nonumber.
\end{align*}

Direct computation provides with
\begin{align*}
    &G^sW^s=\begin{pmatrix} G &K_p \end{pmatrix}\frac{M\Psi}{MJ^2+\Psi}\begin{pmatrix}
   1 & - J   \\ 
   -J &  J^2 
\end{pmatrix} = \frac{M\Psi(G-K_pJ)}{MJ^2+\Psi}\begin{pmatrix} 1 &-J \end{pmatrix},
\end{align*}
so by our assumption $G \neq J K_p$ and the case study $M \neq 0$, the pair is stabilizable.

Before proving the other case $M=0$, we rewrite the Riccati equation in \eqref{eq:proof_detecta_Ricc} as
\begin{align*}
    (F^{s})^2\hat{\Sigma}-\hat{\Sigma}+(G^s)^2W^s-\frac{(F^s\hat{\Sigma}(H+J\overline{\Gamma}))^2}{(H+J\overline{\Gamma})^2\hat{\Sigma}+ G^2 M +\Psi}=0. 
\end{align*}
When $M=0$, it has two solutions $\hat \Sigma =0$ and $\hat \Sigma_M  = \frac{\Psi((F^{s})^2-1)}{(H+J\overline{\Gamma})^2}$, and we need to show that the recursion converges to the maximal solution $\hat \Sigma_M$.

The Riccati recursion when $M=0$ is given by
\begin{align}\label{eq:M_rec}
\hat{\Sigma}_{i+1}=(F^{s})^2\hat{\Sigma}_i\left(1-\frac{\hat{\Sigma}_i(H+J\overline{\Gamma})^2}{(H+J\overline{\Gamma})^2\hat{\Sigma}_i+ \Psi}\right)=\frac{(F^{s})^2\hat{\Sigma}_i\Psi}{(H+J\overline{\Gamma})^2\hat{\Sigma}_i+ \Psi},
\end{align}
and one can show directly that it is increasing for $0< \hat \Sigma_i <\hat \Sigma_M$, and it is decreasing for $\hat \Sigma_M<\hat \Sigma_i$. The proof is completed by noting that choosing $M_1>0$ implies 
$$\hat{\Sigma}_2 =(G^s)^2W^s = \frac{M_1 \Psi (G-K_pJ)^2}{M_1J^2+\Psi}>0.$$
\end{proof}

\section{Proof of Technical lemmas}\label{sec:proofs}

In this section, we provide proofs of Lemmas 1 - 4 consecutively.

\begin{proof}[Proof of lemma \ref{lemma:Gaussian_linear_policy}] 
A general policy of the optimization in \eqref{eq:n_letter_capacity} is characterized by the causal conditioning $P(\vx^n ||\vy^n)$, which is equivalently represented by $\{P(\vx_i | \vy^{i-1},\vx^{i-1})\}_{i=1}^n$. We show that for any policy of this form, one can construct a linear policy of the form 
$$\vx_i= \Gi \sd \Err{i} + \Gti \dhat{\vs}_i + \mathbf{m}_i, \ \ i=1, \dots ,n$$ that achieves at least the same objective while satisfying the LQR control constraint. Thus, linear policies of the form in \eqref{eq:linear_policy} suffice to achieve the optimum in \eqref{eq:n_letter_capacity}. 

First, we restrict the optimization to Gaussian policies. For any policy $P(\vx_i|\vx^{i-1},\vy^{i-1})$, we construct a feasible Gaussian policy $\mathbb{G}(\vx_i|\vx^{i-1},\vy^{i-1})$ that achieves at least the same objective. To construct the Gaussian policy, first compute the joint distribution of the original policy by the chain rule 
\begin{align}\label{eq:joint_dist}
    P(\vx^n, \vy^n, \hat{\vs}^{n+1}) &= P(\hat{\vs}_1) \prod_{i=1}^n P(\hat{\vs}_{i+1},\vy_i|\vx_i, \hat{\vs}_i) P(\vx_i|\vx^{i-1}, \vy^{i-1})
\end{align}
where $P(\hat{\vs}_1)$ is given by the setup. We define a jointly Gaussian measure $\mathbb{G}(\vx^n, \vy^n, \hat{\vs}^{n+1})$ that have the first and second moments of \eqref{eq:joint_dist}, and then compute the induced conditional distribution $\mathbb{G}(\vx_i|\vx^{i-1},\vy^{i-1})$.

Note that $\hat{\vs}^{i}$ is a function of  $(\vx^{i-1}, \vy^{i-1})$ and therefore is omitted from the conditioning in the distribution of $\vx_i$. In addition, the distributions $ P(\hat{\vs}_{i+1}|\vx_i, \hat{\vs}_i, \vy_i), P(\vy_i|\vx_i, \hat{\vs}_i)$ are given by the model in \eqref{eq:encoder_model} and do not depend on the optimization variables, and they are Gaussian, which makes the joint distribution jointly Gaussian if the policy terms $P(\vx_i|\vx^{i-1}, \vy^{i-1})$ are Gaussian.

In \cite{tanaka_allerton}, it is shown that replacing any input policy by this Gaussian policy, the joint distribution $\mathbb{G}(\vx^n, \vy^n, \hat{\vs}^{n+1})$ induced from this Gaussian policy, is jointly Gaussian and its first and second moments are equal to the first and second moments of the joint distribution $P(\vx^n, \vy^n, \hat{\vs}^{n+1})$. Therefore the control-cost constraint in \eqref{eq:intro_LQR_cost} holds under the new policy because it only depends on the covariances of the random sequences $\{\vx_i\}_{i=1}^n$, $\{\hat{\vs}_i\}_{i=1}^{n+1}$. For the directed information, we have  
\begin{align*}
    I_P(\vx^n \rightarrow \vy^n)&=h_P(\vy^n)-h_{\mathbb{G}}(\vy^n||\vx^n) \\ &\leq h_{\mathbb{G}}(\vy^n)-h_{\mathbb{G}}(\vy^n||\vx^n) = I_{\mathbb{G}}(\vx^n \rightarrow \vy^n)
\end{align*}
where the first equality 
is since $h_P(\vy^n||\vx^n)$ depends on the Gaussian dynamical system and does not depend on the policy, and the inequality follows from the fact that a Gaussian distribution maximizes the differential entropy when the first and second moments are fixed. 

For a Gaussian policy $\mathbb{G}(\vx^n ||\vy^n)$, we construct a new policy with a Gaussian measure  $\mathbb{Q}$ of the form 
\begin{align}\label{eq:linear_policy_d}
    \vx_i&= \Gi \sd \Err{i} + \Gti N_i^{\dagger} \dhat{\vs}_i + \mathbf{m}_i,
\end{align}
where $\Gi \triangleq \E_\mathbb{G} [\vx_i \Err{i}^T] $, $\Gti \triangleq \E_\mathbb{G} [\vx_i \dhat{\vs}_i^T]$, $\hat{\Sigma}_i =\mathbf{cov}_\mathbb{G}\Err{i}$, $N_i = \mathbf{cov}_\mathbb{G}(\dhat{\vs}_i)$ and $\mathbf{m}_i$ is independent of $(\vx^{i-1}, \vy^{i-1})$ and is distributed as $\mathbf{m}_i \sim \mathcal{N}(0, M_i)$ with
\begin{align}
    M_i \triangleq \E_\mathbb{G}[\vx_i \vx_i^T] - \E_\mathbb{G} [\vx_i \Err{i}^T] \sd \E_\mathbb{G} [\Err{i} \vx_i^T] -\E_\mathbb{G} [\vx_i \dhat{\vs}_i^T] N_i^{\dagger} \E_\mathbb{G} [\dhat{\vs}_i\vx_i^T].
\end{align} 

Recall that the capacity of a system is the maximum of the directed information $I(\vx^n \rightarrow \vy^n)$ under the control cost constraint $\mathcal{J}_n \leq p$, the main idea here is to show that the control cost and the directed information are the same under G and Q. We show that the directed information only depends on the covariance of $\psi_Y$.
\begin{align*}
    I(\vx^n \rightarrow \vy^n) &= \sumT h(\vy_i|\vy^{i-1}) - h(\vy_i|\vx^i,\vy^{i-1})\\
    &\stackrel{(a)}{=} \sumT h(\vy_i-\E[\vy_i|\vy^{i-1}]|\vy^{i-1}) - h(\vy_i-\E[\vy_i|\vx^i,\vy^{i-1}]|\vx^i,\vy^{i-1})\\&\stackrel{(b)}{=} \sumT h(\mathbf{\psi}_i) - h(\mathbf{e}_i)\\
    &\stackrel{(c)}{=} \frac{1}{2} \sumT  \left(\log \det \Psi_{Y,i} - \log \det \Psi\right),
\end{align*}
where $(a)$ follows from the fact that the estimates are deterministic functions of the their corresponding past measurements, $(b)$ follows from the fact that for an optimal estimator the error is orthogonal to the measurements, and $(c)$ follows from the definitions $\psi_i \triangleq \vy_i - \E[\vy_i|\vy^{i-1}] = \vy_i-(H + J \Gti)\dhat{\vs}_i\sim \mathcal{N}(0, \Psi_{Y,i})$ as defined under \eqref{eq:dhat_discussion} and $\mathbf{e}_i\sim \mathcal{N}(0, \Psi)$ as defined in \eqref{eq:primary_err}.

The control cost depends on $\E[\vx_i^TQ\vx_i]$ and $\E[\hat{\vs}_i^TQ\hat{\vs}_i]$, the term
$\E[\hat{\vs}_i^TQ\hat{\vs}_i]$ can be expressed as 
\begin{align*}
    \E[\hat{\vs}_i^TQ\hat{\vs}_i] &=  \E[\Err{i}^TQ\Err{i}]+ \E[\dhat{\vs}_i^T Q \dhat{\vs}_i]\\
    &=  \Tr(\mathbf{cov}\Err{i} Q)+ \Tr(\mathbf{cov}(\dhat{\vs}_i) Q)
\end{align*}
We show by induction that the joint Gaussian distribution of the tuple $\Phi_i \triangleq (\dhat{\vs}_i, \hat{\vs}_i - \dhat{\vs}_i, \vx_i, \psi_i)$ has the same covariance matrix under $\mathbb{G}$ and $\mathbb{Q}$. This implies that both measures achieve the same directed information and control constraint, so we can restrict the optimiztion to inputs of the form \eqref{eq:linear_policy_d}.

For the base case of the induction, $\Phi_1 = (0,\hat{\vs}_1,0, 0)$ has the same covariance matrix under $\mathbb{Q}$ as under $\mathbb{G}$ by the construction of $\vx_1$. For the induction step, assume that $\{\Phi_i\}_{i=1}^t$ has the same covariance matrix under $\mathbb{G}$ and $\mathbb{Q}$. 
First, the encoder estimator $\hat{\vs}_{t+1}$  depends on $\hat{\vs}^t$, and therefore the decoder estimator $\dhat{\vs}_{t+1} = \E [\hat{\vs}_{t+1}|\vy^t] = \E [\hat{\vs}_{t+1}|\mathbf{\psi}^t]$,
where the second equality holds by the innovation approach \cite[Chapter 4]{kailath_booklinear}, and by the induction hypothesis $\psi_t$ have the same distribution under $\mathbb{G}$ and $\mathbb{Q}$. The covariance between $\dhat{\vs}_{t+1}$ and $\Err{t+1}$ is $0$ by the orthogonality principle. For the system input $\vx_i$ it can be verified by \eqref{eq:linear_policy_d} that $\E_\mathbb{Q}[\vx_{t+1}\vx_{t+1}^T] =\E_\mathbb{G}[\vx_{t+1}\vx_{t+1}^T] $ by the construction of $\Gi, \Gti, \mathbf{m}_i$.
For the covariance between $\vx_{t+1}$ and $\Err{t+1}$ we use the orthogonality property of $\Gi$ which means $\Gi(I-\sd \hat{\Sigma}_i) = 0$ (proof below) to show
\begin{align*}
    \E_\mathbb{Q}[\vx_{t+1}\Err{t+1}^T] &= \E_\mathbb{Q}[\Gamma_{1,t+1}\hat{\Sigma}_{t+1}^{\dagger}\Err{t+1}\Err{t+1}^T] \\ &= \Gamma_{1,t+1}\hat{\Sigma}_{t+1}^{\dagger}\hat{\Sigma}_{t+1} \\ &= \Gamma_{1,t+1}\\& = \E_\mathbb{G}[\vx_{t+1}\Err{t+1}^T]
\end{align*}
The same claim holds for the covariance between $\vx_{t+1}$ and $\dhat{\vs}_{t+1}$ due to the orthogonality constraint~$\Gti (I- N_i^{\dagger}N_i) = 0$.

For the covariance of $\psi_{t+1} $ we show that it only depends on $(\vx_{t+1}, \Err{t+1})$ that we already saw that they have the same distribution under $\mathbb{Q}$ as well as under $\mathbb{G}$, by computing $\mathbf{cov}(\psi_{t+1}) $ by the model define in \eqref{eq:encoder_model} regardless to the policy.
\begin{align*}
    \mathbf{cov}(\vy_{t+1} -(H + J \Gamma_{2,t+1})\dhat{\vs}_{t+1}) &\stackrel{(a)}{=} 
    \mathbf{cov}(H \hat{\vs}_{t+1} +J\vx_{t+1} +\mathbf{e}_{t+1} - H \dhat{\vs}_{t+1})  \\
    &\stackrel{(b)}{=}  \mathbf{cov}(J\vx_{t+1} + H\Err{t+1}) +\Psi
\end{align*}
where (a) follows from our model definition in \eqref{eq:encoder_model},  and (b) follows from the independence of the innovation $\mathbf{e}_{t+1}$ and the tuple $(\vx^{t+1}, \{\Err{j}\}_{j=1}^{t+1})$. 

So, since the optimal policy could be as in  $\vx_i= \Gi \sd \Err{i} + \Gti N_i^{\dagger} \dhat{\vs}_i + \mathbf{m}_i$, we can restrict the set of such policies to policies of the form $\vx_i= \Gi \sd \Err{i} + \Gti \dhat{\vs}_i + \mathbf{m}_i$ as required.
\end{proof}

\begin{proof}[Proof of lemma \ref{lemma:computing_estimators}]
Substituting the policy from Lemma \ref{lemma:Gaussian_linear_policy} into the channel model in \eqref{eq:encoder_model}
provides the state space
\begin{align}\label{eq:proof2_statespace}
    \hat{\vs}_{i+1} &= (F+G\Gi \sd) \Err{i} + (F+G \Gti) \dhat{\vs}_i + G\mathbf{m}_i + K_p \mathbf{e}_i \nonumber\\
        \vy_i &= (H+J\Gi \sd) \Err{i} + (H+ J \Gti) \dhat{\vs}_i + J \mathbf{m}_i  + \mathbf{e}_i.
\end{align}
The innovation of the measurements (the system outputs) are 
\begin{align}\label{eq:psi}
    \psi_i&\triangleq \vy_i  - \E[\vy_i|\vy^{i-1}]\nonumber\\
    &= \vy_i-(H+J\Gti)\dhat{\vs}_i \nonumber\\
    &= (H+J\Gi \sd) \Err{i} + J \mathbf{m}_i  + \mathbf{e}_i,
\end{align}
and consequently we can compute its variance $\Psi_{Y,i}=\mathbf{cov}(\psi_i)$ as
\begin{align*}
    \Psi_{Y,i} &= (H+J\Gi \sd)\hat{\Sigma}_i (H+J\Gi \sd)^T + J M_i J^T +  \Psi.
\end{align*}


For the model in \eqref{eq:observer_model}, we compute the estimate $\dhat{\vs}_i$ recursively by utilizing the innovations approach that is optimal for linear processes with Gaussian disturbances \cite[Chapter 4]{kailath_booklinear} . The estimates can be computed, using the inner product  $<\vx,\vy>=\E[\vx\vy^T]$, as follows
\begin{align*}
    \dhat{\vs}_{i+1} & = \E[\hat{\vs}_{i+1}| \psi^i] \\
    &=\sum_{j=1}^i <\hat{\vs}_{i+1},\psi_j>\Psi_{Y,j}^{-1}\psi_j \\
    &\stackrel{(a)}=(F + G\Gti) \dhat{\vs}_i + <\hat{\vs}_{i+1},\psi_i>\Psi_{Y,i}^{-1}\psi_i \\
    &\stackrel{(b)}=(F + G\Gti) \dhat{\vs}_i + ((F+G\Gi\sd)\hat{\Sigma}_i(H+J\Gi\sd)^T + GM_iJ^T +K_p\Psi)\Psi_{Y,i}^{-1}\psi_i \\
    &\triangleq (F + G \Gti) \dhat{\vs}_i + K_{Y,i}\psi_i,
\end{align*}
where $(a)$ follows from by substituting $\hat{\vs}_{i+1}$ from \eqref{eq:proof2_statespace} and the fact that $<(F+G\Gi \sd) (\hat{\vs}_i  - \dhat{\vs}_i) + G\mathbf{m}_i + K_p \mathbf{e}_i ,\psi_j>=0$ for $j<i$. Step $(b)$ follows from by substituting $\hat{\vs}_{i+1}, \psi_i$ by the definitions in \eqref{eq:proof2_statespace} and \eqref{eq:psi}, and the fact that $\{\dhat{\vs}_i, \Err{i}, \mathbf{m}_i, \mathbf{e}_i\}$ are mutually independent, along with the notation $K_{Y,i}=((F+G\Gi\sd)\hat{\Sigma}_i(H+J\Gi\sd)^T + GM_iJ^T +K_p\Psi)\Psi_{Y,i}^{-1}$. We proceed to compute the recursion of error covariance 
\begin{align*}
\hat{\Sigma}_{i+1}&=\E[\Err{i+1}\Err{i+1}^T] \\
    &=\E[((F+G\Gi \sd)\Err{i}+G\mathbf{m}_i + K_p \mathbf{e}_i-K_{Y,i}\psi_i)((F+G\Gi \sd)\Err{i}+G\mathbf{m}_i + K_p \mathbf{e}_i-K_{Y,i}\psi_i)^T] \\
    &= (F+G\Gi \sd)\hat{\Sigma}_i(F+G\Gi \sd)^T + GM_iG^T + K_p\Psi K_p^T -K_{Y,i}\Psi_{Y,i}K_{Y,i}^T,
\end{align*}
where the last equality holds since the tuple $(\Err{i}, \mathbf{m}_i, \mathbf{e}_i)$ are mutually independent and $\psi_i =(H+J\Gi \sd)\Err{i} + J\mathbf{m}_i + \mathbf{e}_i$ is computed by substituting  \eqref{eq:proof2_statespace} into \eqref{eq:psi}. 

Finally, we compute the directed information 
\begin{align*}
    I(\vx^n \rightarrow \vy^n) &= \sumT h(\vy_i|\vy^{i-1}) - h(\vy_i|\vx^i,\vy^{i-1})\\
    &\stackrel{(a)}{=} \sumT h(\vy_i-\E[\vy_i|\vy^{i-1}]|\vy^{i-1}) - h(\vy_i-\E[\vy_i|\vx^i,\vy^{i-1}]|\vx^i,\vy^{i-1})\\&\stackrel{(b)}{=} \sumT h(\mathbf{\psi}_i) - h(\mathbf{e}_i)\\
    &\stackrel{(c)}{=} \frac{1}{2} \sumT  \left(\log \det \Psi_{Y,i} - \log \det \Psi\right),
\end{align*}
where $(a)$ follows from the fact that the estimates are deterministic functions of the their corresponding past measurements, $(b)$ follows from the fact that for an optimal estimator the error is independent of the measurements, and $(c)$ follows from the definitions $\psi_i\sim \mathcal{N}(0, \Psi_{Y,i})$ and $\mathbf{e}_i\sim \mathcal{N}(0, \Psi)$.
\end{proof}

\begin{proof}[Proof of lemma \ref{lemma:minimizig_Gamma_2}]
We show that the optimization of the control cost with respect to the sequence $\{\Gamma_{2,i}\}$ can be done using the LQR algorithm. Consequently, the control cost and the objective function only depend on the decision variables $\{M_i, \Gi\}_{i=1}^n$. 

Recall the control cost definition
\begin{align*}
    \mathcal{J}_n=\frac1{n}\left[ \E[\vs_{n+1}^T Q \vs_{n+1}] + \sum_{i=1}^n (\E[\vs_i^TQ\vs_i] + \E[\vx_i^TR\vx_i]) \right].
\end{align*}
We start by computing a single summand of the control cost for some $i$
\begin{align}\label{eq:proof_lemma2_1}
    \E[\vx_i^TR \vx_i] + \E[\vs_{i}^T Q \vs_i] 
    &\stackrel{(a)}= \E[\vx_i^TR \vx_i] + \E[(\vs_i-\hat{\vs}_i)^T Q (\vs_i-\hat{\vs}_i)] +\E [\Err{i}^T Q \Err{i}] + \E[\dhat{\vs}_i^T Q \dhat{\vs}_i] \nonumber \\
    &=\E[\vx_i^TR \vx_i] + \E[\dhat{\vs}_i^T Q \dhat{\vs}_i] + \Tr (\Sigma Q) + \Tr (\hat{\Sigma}_i Q)\nonumber\\
    &\stackrel{(b)}= \Tr(M_i R) + \Tr(\Gamma_{1,i}\hat{\Sigma}^\dagger_i\Gamma_{1,i}^T R) + \E[\mathbf{u}_i^TR \mathbf{u}_i] + \E[\dhat{\vs}_{i}^T Q \dhat{\vs}_{i}] + \Tr (\Sigma Q) + \Tr (\hat{\Sigma}_i Q)
\end{align}
where $(a)$ follows from the fact that $\vs_{i} = (\vs_i-\hat{\vs}_i) + \hat{\vs}_i = (\vs_i-\hat{\vs}_i) + \Err{i} + \dhat{\vs}_i $ and the orthogonality of $(\vs_i-\hat{\vs}_i) , \hat{\vs}_i$, since $(\vs_i-\hat{\vs}_i)$ is the error of the estimator $\hat{\vs}_i$, and the orthogonality of $\Err{i} , \dhat{\vs}_i$. Step $(b)$ follows from the structure of system inputs and $\mathbf{u}_i \triangleq \Gamma_{2,i}\dhat{\vs}_i$.

By \eqref{eq:proof_lemma2_1}, we can write the cost as follows:
\begin{align}
    \mathcal{J}_n&=\frac1{n}\left[ \E[\dhat{\vs}_{n+1}^T Q \dhat{\vs}_{n+1}] + \sumT (\E[\dhat{\vs}_i^TQ \dhat{\vs}_i] + \E[\mathbf{u}_i^TR\mathbf{u}_i]) \right] \nonumber\\
    & \label{eq:lqr_rest}+\frac1{n}\sum_{i=1}^n \left[\Tr( (M_i +\Gamma_{1,i}\hat{\Sigma}^\dagger_i\Gamma_{1,i}^T) R)  \right] + \frac1{n}\sum_{i=1}^{n+1}\left[ \Tr (\Sigma Q) + \Tr (\hat{\Sigma}_i Q)\right].
\end{align}
Note that the second and third sums do not depend on $\{\Gamma_{2,i}\}$. The first sum appears as a standard LQR cost along with the dynamical system equation in \eqref{eq:dhat_discussion}
\begin{align*}
    \dhat{\vs}_{i+1} &= F \dhat{\vs}_i +G\Gamma_{2,i}\dhat{\vs}_i + K_{Y,i}\psi_i.
\end{align*}

By the LQR algorithm the optimum of 
\begin{align}\label{eq:LQR_minimum}
    \min_{\{\Gamma_{2,i}\}_{i=1}^n}\frac1{n}\left[ \E[\dhat{\vs}_{n+1}^T Q \dhat{\vs}_{n+1}] + \sum_{i=1}^n (\E[\dhat{\vs}_i^TQ\dhat{\vs}_i] + \E[\mathbf{u}_i^TR\mathbf{u}_i]) \right]
\end{align}
is attained at $$\Gamma_{2,i} = - \KLQRi, \ \ \ i=1,\dots,n$$ where $\KLQRi =\Psi_{\mathtt{LQR},i}^{-1}G^T E_{i+1} F $ and $\Psi_{\mathtt{LQR},i} =R+G^T E_{i+1} G $
and $E_i$ is defined recursively as 
\begin{align}
    E_i = F^TE_{i+1}F+Q - \KLQRi^T \Psi_{\mathtt{LQR},i}  \KLQRi 
\end{align}
with the terminal condition $E_{n+1} = Q$. The LQR algorithm also provides us with the optimal cost of the  \eqref{eq:LQR_minimum}, and it is equal to $\frac1{n} \Tr(\mathbf{cov}(\dhat{\vs}_1) E_1) + \frac{1}{n}\sumT  \Tr(K_{Y,i} \Psi_{Y,i}K_{Y,i}^TE_{i+1})$.


To summarize, the control cost after optimizing $\{
\Gamma_{2,i}\}_{i=1}^n$ is 
\begin{align}
    &\frac1{n} \Tr(\mathbf{cov}(\dhat{\vs}_1) E_1) + \frac{1}{n}\sumT  \Tr(K_{Y,i} \Psi_{Y,i}K_{Y,i}^TE_{i+1})+\frac1{n}\sum_{i=1}^n \left[\Tr((M_i +\Gamma_{1,i}\hat{\Sigma}^\dagger_i\Gamma_{1,i}^T ) R) \right] \nonumber\\
    &\ \ + \frac1{n}\sum_{i=1}^{n+1}\left[ \Tr (\Sigma Q) + \Tr (\hat{\Sigma}_i Q)\right],
\end{align}
and it is a function of $\{\Gamma_{1,i},M_i\}_{i=1}^n$.

 \end{proof}

\begin{proof}[Proof of lemma \ref{lemma:T-letter_SDP}]
The control cost is given by
\begin{align}
    \mathcal J_n& = \frac1{n} \Tr(\mathbf{cov}(\dhat{\vs}_1) E_1) + \frac{1}{n}\sumT  \Tr(K_{Y,i} \Psi_{Y,i}K_{Y,i}^TE_{i+1}) + \frac1{n}\sum_{i=1}^n \left[\Tr((M_i +\Gamma_{1,i}\hat{\Sigma}^\dagger_i\Gamma_{1,i}^T ) R) \right] \nonumber\\
    &\ \ + \frac1{n}\sum_{i=1}^{n+1}\left[ \Tr (\Sigma Q) + \Tr (\hat{\Sigma}_i Q)\right].
\end{align}
Our objective is to write the control cost as a linear function of the decision variables. An evident challenge is the sum  $\frac{1}{n}\sumT  \Tr(K_{Y,i} \Psi_{Y,i}K_{Y,i}^TE_{i+1})$ since $K_{Y,i}$ is a rational function of the decision variables. We rewrite the sum using the Riccati recursion in \eqref{eq:cov_decoder_err} as follows:
\begin{align}\label{eq:KalmanGain_Linear}
    &\sumT  \Tr(K_{Y,i} \Psi_{Y,i}K_{Y,i}^TE_{i+1})\} \nonumber\\
    &= \sumT \Tr(( - \hat{\Sigma}_{i+1} + (F+G\Gi \sd)\hat{\Sigma}_i(F+G\Gi \sd)^T + GM_iG^T + K_p\Psi K_p^T)E_{i+1}) \nonumber\\
    &= \sumT \Tr(( F\hat{\Sigma}_i F^T  - \hat{\Sigma}_{i+1})E_{i+1}) + \Tr( (M_i + \Gi \sd \Gi^T) G^TE_{i+1}G) + \Tr(K_p\Psi K_p^TE_{i+1}) \nonumber\\
    &\ \ +2\Tr(G \Gi F^TE_{i+1})  \nonumber\\
    &\stackrel{(a)}= \sumT \left[\Tr(\hat{\Sigma}_i (-Q + \KLQRi^T(R+G^TE_{i+1}G)\KLQRi))  + \Tr( (M_i + \Gi \sd \Gi^T) G^TE_{i+1}G) \right.\nonumber\\
    &\ \ + \left.\Tr(K_p\Psi K_p^TE_{i+1})  +2\Tr(G \Gi F^TE_{i+1})\right] + \Tr(\hat{\Sigma}_1E_1 - \hat{\Sigma}_{n+1}E_{n+1}) ,
\end{align}
where in $(a)$ we use the control Riccati recursion \eqref{eq:control_recursion}.

Substituting \eqref{eq:KalmanGain_Linear} in the control cost, we obtain
\begin{align}\label{eq:cost_almost_linear}    \mathcal J_n& = \frac1{n}\sum_{i=1}^n \left[ \Tr( \hat{\Sigma}_i\KLQRi^T \Psi_{\mathtt{LQR},i}\KLQRi) + \Tr((M_i +\Gamma_{1,i}\hat{\Sigma}^\dagger_i\Gamma_{1,i}^T ) \Psi_{\mathtt{LQR},i}) +\Tr(K_p\Psi K_p^TE_{i+1}) \right]\nonumber\\
    &\ \ + 2\sumT\Tr(G \Gi F^TE_{i+1}) + \Tr (\Sigma Q) + \mathcal{E}_n,
\end{align}
with $\mathcal{E}_n\triangleq\frac1{n}\left(\Tr (( \Sigma + \hat{\Sigma}_{n+1}) Q) + \Tr(\mathbf{cov}(\dhat{\vs}_1) E_1) + \Tr(\hat{\Sigma}_1E_1 - \hat{\Sigma}_{n+1}E_{n+1})\right)$.

We proceed to formulate the optimization problem as a SCOP where all the constraints are represented by LMIs and traces of a linear combination of the decision variables.
Notice that even though the Trace in \eqref{eq:cost_almost_linear} is not linear in the variables $(\Gi, \hat{\Sigma}_i)$, we can define
\begin{align}\label{eq:Pi}
    \Pi_i \triangleq \Gi \sd \Gi^T + M_i
\end{align}
to rewrite the cost as a linear function of the variables  $\Pi_i, \hat{\Sigma}_i, \Gi$ as
\begin{align*}
     \mathcal{J}_n  
 &= \frac1{n}\sum_{i=1}^n \left[ \Tr( \hat{\Sigma}_i\KLQRi^T \Psi_{\mathtt{LQR},i}\KLQRi) + \Tr(\Pi_i  \Psi_{\mathtt{LQR},i}) +\Tr(K_p\Psi K_p^TE_{i+1}) \right]\nonumber\\
    &\ \ + 2\frac1{n}\sumT \Tr(\Gi \KLQRi^T \Psi_{\mathtt{LQR},i}) + \Tr (\Sigma Q) + \mathcal{E}_n,
\end{align*}
where we also used the identity $\KLQRi = \Psi_{\mathtt{LQR},i}^{-1} G^TE_{i+1}F$.

Similarly, the objective in \eqref{eq:objective} is expressed as the log-determinant of $\Psi_{Y,i}$, which is not a linear function of the decision variables. With \eqref{eq:Pi}, we can write $\Psi_{Y,i}$ as
    \begin{align*}
        \Psi_{Y,i}&= (H+J\Gi \sd)\hat{\Sigma}_i (H+J\Gi \sd)^T + JM_iJ^T + \Psi \\
        &= H \hat{\Sigma}_i H^T + J \Gi \sd \hat{\Sigma}_i H^T + H \hat{\Sigma}_i  \sd \Gamma_{1,i}^T J^T + J \Gi \hat{\Sigma}_i \sd \Gi^T J^T + JM_iJ^T + \Psi \\ 
        &= J\Pi_i J^T +H\hat{\Sigma}_i H^T+H \Gi^TJ^T + J\Gamma_i H^T+ \Psi,
    \end{align*}
where the last equality follows from the orthogonality constraint in \eqref{eq:orthogonal} and \eqref{eq:Pi}. Note that $\Psi_{Y,i}$ is now a linear function of $\Pi_i, \hat{\Sigma}_i, \Gi$.

In a similar way we can write $K_{Y,i}$ as a linear function of  $\Pi_i, \hat{\Sigma}_i, \Gi$
\begin{align*}
    K_{Y,i}&= ((F+G\Gi\sd)\hat{\Sigma}_i(H+J\Gi\sd)^T + GM_iJ^T +K_p\Psi)\Psi_{Y,i}^{-1} \\
    &=(F\Gi J^T + F \hat{\Sigma}_i H^T + G\Pi_i J^T + G \Gi H^T +K_{p}\Psi) \Psi_{Y,i}^{-1}.
\end{align*}
    
    The relationship between $\Pi_i$ and the other variables in \eqref{eq:Pi} can be relaxed to an LMI constraint. In particular, we can eliminate the variable $M_i \succeq 0$ to obtain $\Pi_i - \Gi \sd \Gi ^T \succeq 0$. Using the Schur complement, we obtain the LMI
    \begin{align*}
         \begin{pmatrix}
        \Pi_i & \Gi\\  \Gi^T & \hat{\Sigma}_i
    \end{pmatrix} \succeq 0,
    \end{align*}
    which is equivalent to $\hat{\Sigma}_i \succeq 0 , \Pi_i - \Gi \sd \Gi ^T \succeq 0$ and the orthogonality constraint.
    


    Finally, the Riccati equation in \eqref{eq:cov_decoder_err} is relaxed to a Riccati inequality 
    \begin{align*}
        \hat{\Sigma}_{i+1} \preceq F\hat{\Sigma}_i F^T+F\Gi G^T+ G\Gi F^T + G\Pi_i G^T + K_p\Psi K_p^T -K_{Y,i}\Psi_{Y,i}K_{Y,i}^T.
    \end{align*}
    Using the Schur complement transformation, we can write it as
    \begin{align*}
        \begin{pmatrix}
     F\hat{\Sigma}_i F^T+F\Gi G^T+ G\Gi F^T + G\Pi_i G^T+  K_{p}\Psi K_{p}^T - \hat{\Sigma}_{i+1}& K_{Y,i}\Psi_{Y,i}\\ 
     \Psi_{Y,i} K_{Y,i}^T & \Psi_{Y,i}
    \end{pmatrix} \succeq 0.
    \end{align*}

     
\end{proof}

\section{Conclusion}\label{sec:conclusion}
We studied communication via LQG control systems and derived a finite-dimensional convex converse (upper bound) on capacity. The bound recovers classical limits and is exact for scalar plants; for vector plants we provided a Riccati-based sufficient condition, and numerical experiments indicate tightness across all tested instances. We also extracted policy parameters $(\Gamma,M)$ from the upper bound to yield matching lower bounds (empirically). In the scalar case when $M=0$, there exists a concrete coding scheme that achieves the capacity. Naturally, we aim in future work to establish the tightness of the upper bound for vector systems, and to construct explicit coding scheme when $M\neq0$.

\appendices
\section{Regularity conditions}
\label{app:regularity_conds}
In this appendix, we provide some preliminaries on dynamical systems. We start with basic definitions.
\begin{definition}
    The pair $(F, H)$ is detectable if there exists a matrix $K$ such that $\rho (F-KH) < 1$.
\end{definition}
\begin{definition}
    The pair $(F, W)$ is controllable on the unit circle, if for any $\vx$ and $\lambda$ such that $\vx^TF=\vx^T\lambda$ if $|\lambda| = 1$, then $\vx W \neq 0$.
\end{definition}
\begin{definition}
    The pair $(F, W)$ is stabilizable, if for any $\vx$ and $\lambda$ such that $\vx^TF=\vx^T\lambda$ if $|\lambda| \ge1$, then $\vx W \neq 0$.
\end{definition}

\subsection{The stabilizing solution for $\Sigma$}\label{app:sub_sigma}
We discuss the convergence of the recursion in \eqref{eq:cov_encoder_err} to the stabilizing solution of its corresponding Riccati equation (see \eqref{eq:riccati_sigma})
\begin{align*}
    \Sigma = F\Sigma F^T +W - K_p\Psi K_p^T,
\end{align*}
where $K_p = (F\Sigma H^T +L)\Psi^{-1}$ and $\Psi = H\Sigma H^T + V$.
If the pair $(F, H)$ is detectable and $(F^s, GW^s)$ is controllable on the unit circle where $F^s:=F-LV^{-1}H$ and $W^s:=W-LV^{-1}L^T$, there exists a unique stabilizing solution to the Riccati equation,  i.e. $\Sigma_F$ is stabilizing solution if it solves the Riccati equation and  $F-K_p\Sigma_s H$ is stable. We proceed to discuss the convergence of the recursion to the stabilizing solution. If the pair $(F^s, GW^s)$ is stabilizable, then the stabilizing solution is the unique PSD solution to the Riccati equation \cite[app E]{kailath_booklinear}. This property guarantees that for all initial points $\Sigma_1$, the recursion converges to the stabilizing solution. Alternatively, we can assume that the initial point satisfies $\Sigma_1 \succeq \Sigma_s$ where $\Sigma_s$ is the stabilizing solution \cite{convergence_detectable_initial}. We remark that we cover here only two methods to guarantee convergence to stabilizing solution. 


\subsection{LQG control}\label{app:sec_lqg}
For the recursion in \eqref{eq:control_recursion} to converge to the stabilizing solution of its Riccati equation as in \eqref{eq:lqr_Riccati_equation} 
\begin{align*}
    E = F^TE F + Q - K_{\mathtt{LQR}}^T \Psi_{\mathtt{LQR}} K_{\mathtt{LQR}},
\end{align*}
where $K_{\mathtt{LQR}} = \Psi_{\mathtt{LQR}}^{-1}G^TEF$, $\Psi_{\mathtt{LQR}} = R+G^TEG$, we assume the standard regularity conditions, stabilizability of $(F,G)$, and the stabilizability of $(F^T, Q)$. Therefore, $Q\succ0$ is sufficient for the stabilizability of $(F^T, Q)$.

\subsection{Stabilizing solution for $\hat \Sigma$ - Theorem \ref{theorem:achievabilty}}\label{app:sub_hat_sigma}
Following the discussion in Section \ref{app:sub_sigma}, there are several sets of assumptions to guarantee the convergence of $\hat\Sigma_i$ to the stabilizing solution of the Riccati equation 
\begin{align*}
    \hat{\Sigma}= (F+G \overline{\Gamma} )\hat{\Sigma}(F+G \overline{\Gamma})^T + GMG^T + K_p\Psi K_p^T -\overline{K}_{Y}\overline{\Psi}_{Y}\overline{K}_{Y}^T
\end{align*}
where $\overline{K}_Y=  ((F+G \overline{\Gamma} )\hat{\Sigma}(H+J\overline{\Gamma})^T + GMJ^T +K_p\Psi)\overline{\Psi}_Y^{-1}$ and $\overline{\Psi}_Y = (H+J\overline{\Gamma})\hat{\Sigma} (H+J\overline{\Gamma} )^T + J M J^T +  \Psi$. 

In Theorem \ref{theorem:achievabilty}, we assumed the stronger assumptions that $(F + G \overline{\Gamma}, H+J\overline{\Gamma})$ is detectable, and the pair $(F^s, G^s W^{s})$ is stabilizable, where 
\begin{align}
    F^s &= F+G\overline{\Gamma} -(GMJ^T+K_p\Psi)(JMJ^T + \Psi)^{-1}(H+J\overline{\Gamma} )^T\\
    W^s &= \begin{pmatrix}
  M & 0\\ 
  0 & \Psi
\end{pmatrix} - 
\begin{pmatrix}
  MJ^T \\ 
  \Psi
\end{pmatrix} (JMJ^T+\Psi)^{-1}
\begin{pmatrix}
  JM &\Psi
\end{pmatrix} \nonumber\\
    G^s&=\begin{pmatrix} G &K_p \end{pmatrix}.\nonumber
\end{align}
Under these assumptions, there exists a unique PSD solution to the Riccati equation, and if the sufficient condition in \eqref{eq:main_ricc_eq_sufficient} holds then the lower bound and the upper bound coincide. We remark that these regularity conditions are not necessary - for instance, in the proof of Theorem \ref{theorem:scalar} on the scalar case, there is a scenario where stabilizability does not hold and we show the convergence directly. 


\bibliographystyle{unsrt}
\bibliography{ref}

\begin{thebibliography}{10}

\bibitem{WuChenGunduz2025Actions}
G.~Chen H.~Wu and D.~G{\"u}nd{\"u}z.
\newblock Actions speak louder than words: Rate-reward trade-off in {M}arkov {D}ecision {P}rocesses.
\newblock In {\em Proceedings of the International Conference on Learning Representations (ICLR)}, Singapore, 2025.

\bibitem{Massey90}
J.~Massey.
\newblock Causality, feedback and directed information.
\newblock {\em Proc. Int. Symp. Inf. Theory Applic. (ISITA-90)}, pages 303--305, Nov. 1990.

\bibitem{Kramer98}
G.~Kramer.
\newblock {\em Directed information for channels with feedback}.
\newblock Ph.{D}. dissertation, Swiss Federal Institute of Technology (ETH) Zurich, 1998.

\bibitem{CoverPombra}
T.~M. {Cover} and S.~{Pombra}.
\newblock {G}aussian feedback capacity.
\newblock {\em IEEE Trans. Inf. Theory}, 35(1):37--43, 1989.

\bibitem{sabag22isit}
O.~Sabag, V.~Kostina, and B.~Hassibi.
\newblock Feedback capacity of {G}aussian channels with memory.
\newblock In {\em 2022 IEEE Int. Symp. Inf. Theory (ISIT)}, pages 2547--2552. IEEE, 2022.

\bibitem{PermuterWeissmanGoldsmith09}
H.H. Permuter, T.~Weissman, and A.~J. Goldsmith.
\newblock Finite state channels with time-invariant deterministic feedback.
\newblock {\em IEEE Trans. Inf. Theory}, 55(2):644--662, Feb. 2009.

\bibitem{caverly2019lmi}
R.~Caverly and J.~R. Forbes.
\newblock {LMI} properties and applications in systems, stability, and control theory, 2019.
\newblock Available at arxiv.org/abs/1903.08599.

\bibitem{tanaka_allerton}
A.~R. Pedram and T.~Tanaka.
\newblock Some results on the computation of feedback capacity of {G}aussian channels with memory.
\newblock In {\em 2018 56th Annual Allerton Conference on Communication, Control, and Computing (Allerton)}, pages 919--926. IEEE, 2018.

\bibitem{Tanaka_SDP_1}
T.~{Tanaka}, K.~K. {Kim}, P.~A. {Parrilo}, and S.~K. {Mitter}.
\newblock Semidefinite programming approach to {G}aussian sequential rate-distortion trade-offs.
\newblock {\em IEEE Trans. Autom. Control}, 62(4):1896--1910, 2017.

\bibitem{Tanaka_SDP_2}
T.~{Tanaka}, P.~M. {Esfahani}, and S.~K. {Mitter}.
\newblock {LQG} control with minimum directed information: Semidefinite programming approach.
\newblock {\em IEEE Trans. Autom. Control}, 63(1):37--52, Jan. 2018.

\bibitem{Gattami}
A.~{Gattami}.
\newblock Feedback capacity of {G}aussian channels revisited.
\newblock {\em IEEE Trans. Inf. Theory}, 65(3):1948--1960, Mar. 2019.

\bibitem{SabagKostinaHassibiMIMO}
V.~Kostina O.~Sabag and B.~Hassibi.
\newblock Feedback capacity of mimo gaussian channels.
\newblock {\em IEEE Trans. Inf. Theory}, 69(10):6121--6136, 2023.

\bibitem{Elia_MIMO_ITW}
A.~Rawat and N.~Elia.
\newblock Feedback capacity of {ISI} {MIMO} channel with colored noise.
\newblock In {\em 2020 IEEE Inf. Theory Workshop (ITW)}, pages 1--5, 2021.

\bibitem{Butman69}
S.A. Butman.
\newblock A general formulation of linear feedback communication systems with solutions.
\newblock {\em IEEE Trans. Inf. theory}, 15:392--400, May 1969.

\bibitem{Butman_conjecture}
S.~{Butman}.
\newblock Linear feedback rate bounds for regressive channels.
\newblock {\em IEEE Trans. Inf. Theory}, 22(3):363--366, 1976.

\bibitem{TiernanSchalk_AR_UB}
J.~{Tiernan} and J.~{Schalkwijk}.
\newblock An upper bound to the capacity of the band-limited {G}aussian autoregressive channel with noiseless feedback.
\newblock {\em IEEE Trans. Inf. Theory}, 20(3):311--316, 1974.

\bibitem{Ebert}
P.~Ebert.
\newblock {The capacity of the {G}aussian channel with feedback }.
\newblock {\em Bell Syst. Tech. J.}, pages 1705--1712, 1970.

\bibitem{Ordentlich96}
E.~Ordentlich.
\newblock On the factor-of-two bound for {G}aussian multiple-access channels with feedback.
\newblock {\em IEEE Trans. Inf. Theory}, 42:2231--2235, 1996.

\bibitem{YoulaCoding}
C.~{Li} and N.~{Elia}.
\newblock Youla coding and computation of {G}aussian feedback capacity.
\newblock {\em IEEE Trans. Inf. Theory}, 64(4):3197--3215, 2018.

\bibitem{9611495}
Charalambos~D. Charalambous, Christos Kourtellaris, and Stelios Louka.
\newblock Sequential characterizations of cover and pombra gaussian feedback capacity: Generalizations to mimo channels via sufficient statistic.
\newblock In {\em 2021 IEEE Information Theory Workshop (ITW)}, pages 1--6, 2021.

\bibitem{Han_GaussianFeedback}
T.~{Liu} and G.~{Han}.
\newblock Feedback capacity of stationary {G}aussian channels further examined.
\newblock {\em IEEE Trans. Inf. Theory}, 65(4):2492--2506, Apr. 2019.

\bibitem{elia_bodemeets}
N.~Elia.
\newblock When {B}ode meets {S}hannon: control-oriented feedback communication schemes.
\newblock {\em IEEE Trans. Autom. Control}, 49(9):1477--1488, Sep. 2004.

\bibitem{Ozarow90}
L.~H. Ozarow.
\newblock Random coding for additive gaussian channels with feedback.
\newblock {\em IEEE Trans. on Info. Theory}, 36(1):17--22, 1990.

\bibitem{Kim10_Feedback_capacity_stationary_Gaussian}
Y.-H. Kim.
\newblock Feedback capacity of stationary {G}aussian channels.
\newblock {\em IEEE Trans. Inf. Theory.}, 56(1):57--85, Jan. 2010.

\bibitem{comments_kim}
M.~S. Derpich and J.~Østergaard.
\newblock Comments on “feedback capacity of stationary gaussian channels”.
\newblock {\em IEEE Trans. Inf. Theory}, pages 1--1, 2022.

\bibitem{CC_discrete_IT_bambos}
C.~K Kourtellaris and C.~D Charalambous.
\newblock Information structures for feedback capacity of channels with memory and transmission cost: Stochastic optimal control and variational equalities.
\newblock {\em IEEE Trans. Inf. Theory}, 64(7):4962--4992, 2017.

\bibitem{CC_LQG_IT_bambos}
C.~D. Charalambous, C.~K. Kourtellaris, and S.~Loyka.
\newblock Capacity achieving distributions and separation principle for feedback gaussian channels with memory: the lqg theory of directed information.
\newblock {\em IEEE Trans. Inf. Theory}, 64(9):6384--6418, 2018.

\bibitem{Permuter06_trapdoor_submit}
H.~H. Permuter, P.~W. Cuff, B.~Van-Roy, and T.~Weissman.
\newblock Capacity of the trapdoor channel with feedback.
\newblock {\em IEEE Trans. Inf. Theory}, 54:3150--3165, 2008.

\bibitem{huleihel2023capacity}
Bashar Huleihel, Oron Sabag, Haim~H Permuter, and Victoria Kostina.
\newblock Capacity of finite-state channels with delayed feedback.
\newblock {\em IEEE Trans. Inf. Theory}, 70(1):16--29, 2023.

\bibitem{huleihel2025duality}
Z.~Aharoni B.~Huleihel, O.~Sabag and H.~H. Permuter.
\newblock The duality upper bound for finite-state channels with feedback.
\newblock {\em IEEE Trans. Inf. Theory}, 2025.

\bibitem{Sabag_BEC}
O.~Sabag, H.H. Permuter, and N.~Kashyap.
\newblock The feedback capacity of the binary erasure channel with a no-consecutive-ones input constraint.
\newblock {\em IEEE Trans. Inf. Theory}, 62(1):8--22, Jan 2016.

\bibitem{Sabag_BIBO_IT}
O.~Sabag, H.~H. Permuter, and N.~Kashyap.
\newblock Feedback capacity and coding for the {BIBO} channel with a no-repeated-ones input constraint.
\newblock {\em IEEE Trans. Inf. Theory}, 64(7):4940--4961, July 2018.

\bibitem{Sabag_UB_IT}
O.~Sabag, H.~H. Permuter, and H.~D. Pfister.
\newblock A single-letter upper bound on the feedback capacity of unifilar finite-state channels.
\newblock {\em IEEE Trans. Inf. Theory}, 63(3):1392--1409, 2017.

\bibitem{Chen05}
J.~Chen and T.~Berger.
\newblock The capacity of finite-state {M}arkov channels with feedback.
\newblock {\em IEEE Trans. Inf. Theory}, 51:780--789, Mar. 2005.

\bibitem{POSTchannel}
H.H. Permuter, H.~Asnani, and T.~Weissman.
\newblock Capacity of a {POST} channel with and without feedback.
\newblock {\em IEEE Trans. Inf. Theory}, 60(10):6041--6057, Oct. 2014.

\bibitem{Eli_NOST_TIT}
O.~Sabag E.~Shemuel and H.~H. Permuter.
\newblock The feedback capacity of noisy output is the state (nost) channels.
\newblock {\em IEEE Trans Inf. Theory}, 68(8):5044--5059, 2022.

\bibitem{atayKostinaposet}
V.~Chandrasekaran E.U.~Atay, E.~Levin and V.~Kostina.
\newblock Poset-markov channels: Capacity via group symmetry, 2025.

\bibitem{charalambous2024siam}
Charalambos~D Charalambous, Christos~K Kourtellaris, and Ioannis Tzortzis.
\newblock Optimal control and signaling strategies of control-coding capacity of general decision models: Applications to gaussian models and decentralized strategies.
\newblock {\em SIAM Journal on Control and Optimization}, 62(1):600--629, 2024.

\bibitem{CharalambousLouka2025JSSC}
C.~D. Charalambous and S.~Louka.
\newblock Signalling and control of nonlinear partially observable stochastic systems: Information states, sufficient statistics, and applications.
\newblock {\em Journal of Systems Science \& Complexity}, 38(1):271--312, February 2025.

\bibitem{kailath_booklinear}
T.~Kailath, A.~H Sayed, and B.~Hassibi.
\newblock {\em Linear estimation}.
\newblock Prentice Hall, 2000.

\bibitem{CC_Bambos_ISIT24}
C.D. Charalambous and S.~Louka.
\newblock Feedback capacity of nonlinear decision models with general noise: {G}aussian applications with filtering and control {R}iccati equations.
\newblock In {\em 2024 IEEE International Symposium on Inf. Theory (ISIT)}, pages 2969--2974, 2024.

\bibitem{bambos_IT_CC_FI}
C.~K Kourtellaris and C.~D Charalambous.
\newblock Information structures for feedback capacity of channels with memory and transmission cost: Stochastic optimal control and variational equalities.
\newblock {\em IEEE Trans. Inf. Theory}, 64(7):4962--4992, 2017.

\bibitem{SchalkwijkKailath66_feedback_scheme}
J.~P.~M. Schalkwijk and T.~Kailath.
\newblock {C}oding scheme for additive noise channels with feedback {I}: No bandwidth constraint.
\newblock {\em IEEE Trans. Inf. Theory}, 12:172--182, 1966.

\bibitem{shayevitz_posterior_mathcing}
O.~Shayevitz and M.~Feder.
\newblock Optimal feedback communication via posterior matching.
\newblock {\em IEEE Trans. Inf. Theory}, 57(3):1186--1222, Mar. 2011.

\bibitem{convergence_detectable_initial}
G.~De~Nicolao and M.~Gevers.
\newblock Difference and differential {R}iccati equations: a note on the convergence to the strong solution.
\newblock {\em IEEE Trans. Autom. Control}, 37(7):1055--1057, 1992.

\end{thebibliography}
\end{document}